\documentclass[11pt]{article}

\usepackage[a4paper]{geometry}

\usepackage[english]{babel}
\usepackage[utf8]{inputenc}
\usepackage[T1]{fontenc}
\usepackage{epsfig, float}
\usepackage{appendix}

\usepackage{mathcomp}
\usepackage{dsfont}
\usepackage{lmodern}
\usepackage{scrextend}
\usepackage{color}

\usepackage[nottoc]{tocbibind}

\usepackage{amsmath,amssymb,amsfonts,amsthm}
\usepackage{amsfonts}
\usepackage{url}
\usepackage{bbm}
\usepackage{mathrsfs}

\usepackage{enumerate}

\newcommand{\norm}[1]{\left\lVert#1\right\rVert}

\newcommand{\R}{\mathbb{R}}
\newcommand{\C}{\mathbb{C}}
\newcommand{\N}{\mathbb{N}}

\renewcommand{\Re}{\operatorname{Re}}
\renewcommand{\Im}{\operatorname{Im}}

\newtheorem{theorem}{Theorem}[section]
\newtheorem{lemma}[theorem]{Lemma}

\newtheorem{remark}[theorem]{Remark}
\newtheorem{notation}[theorem]{Notation}
\newtheorem{definition}[theorem]{Definition}

\numberwithin{equation}{section}

\usepackage[normalem]{ulem}

\date{\today}

\title{\Large{\textsc{One-Boson Scattering Processes in the Massless Spin-Boson Model -- A Non-Perturbative Formula}}}

\author{Miguel Ballesteros\thanks{\texttt{miguel.ballesteros@iimas.unam.mx},
    Instituto de Investigaciones en Matem\'aticas Aplicadas y en Sistemas,
    Universidad Nacional Aut\'anoma de M\'exico}, Dirk-Andr\'e
    Deckert\thanks{\texttt{deckert@math.lmu.de}, Mathematisches Institut
    der Ludwig-Maximilians-Universität München}, Felix
    H\"anle\thanks{\texttt{haenle@math.lmu.de}, Mathematisches Institut
der Ludwig-Maximilians-Universität München}}

\begin{document}
\maketitle

\begin{abstract}
In scattering experiments, physicists observe so-called resonances as peaks at certain energy values in the measured scattering cross sections per solid angle. These peaks are usually associate with certain scattering processes, e.g., emission, absorption, or excitation of certain particles and systems. On the other hand, mathematicians define resonances as poles of an analytic continuation of the resolvent operator through complex dilations. A major challenge is to relate these scattering and resonance theoretical notions, e.g., to prove that the poles of the resolvent operator induce the above mentioned peaks in the scattering matrix. In the case of quantum mechanics, this problem was addressed in numerous works that culminated in Simon's seminal paper \cite{simonnbody} in which a general solution was presented for a large class of pair potentials. However, in quantum field theory the analogous problem has been open for several decades despite the fact that scattering and resonance theories have been well-developed for many models. In certain regimes these models describe  very fundamental phenomena, such as emission and absorption of photons by atoms, from which quantum mechanics originated. In this work we present  a first non-perturbative formula that relates the scattering matrix to the resolvent operator in the massless Spin-Boson model. This result can be seen as a major progress compared to our previous works \cite{bdh-scat} and \cite{bdfh} in which we only managed to derive a perturbative formula.
\end{abstract}


\section{Introduction}
\label{sec:introduction}

In this work we analyze the massless Spin-Boson model which describes  a two-level atom  interacting with a second-quantized massless scalar field.  We  derive a non-perturbative expression of the scattering matrix in terms of the resolvent operator for one-boson processes, and thus, prove an analogous result that was obtained by Simon in \cite{simonnbody} for the N-body Schr\"odinger operator in this particular model of quantum field theory. More precisely, we show that the pole of a meromorphic continuation of the integral kernel of the scattering matrix is located precisely at the resonance energy. The objective in this result is to contribute to the understanding of the relation between resonance and scattering theory. In our previous works \cite{bdh-scat} and \cite{bdfh}, we were already able to derive perturbative results of this kind in case of the massless and massive Spin-Boson models, respectively. However, both results are only given in leading order with respect to the coupling constant. The present work can be seen as a major improvement of these pertubative results because it provides a closed and non-perturbative formula that connects the integral kernel of the scattering matrix elements for one-boson processes in terms of the dilated resolvent.

Our results are based on the well-established fields of scattering and resonance theories and the numerous works in the classical literature of which we want to give a short overview here. Resonance theory, in the realm of quantum field theory, has been developed in a variety of models; see, e.g., \cite{bfs1,bfs2,bfs3,bcfs,bfs100,bbf,fgs100,feshbach,s,f,bffs,pizzo1,pizzo2,bach,bbp,bdh-res}. In these works, several techniques have been invented for  massless models of quantum field theory  in order to cope with the absence of a  spectral gap. Scattering theory has also been developed  in various models of quantum field theory (see, e.g.,
\cite{fau1,fau2,fau3,fgs1,fgs2}) and in particular in the massless Spin-Boson model  (see, e.g., \cite{rgk,rk,rk2,derezinski,bkz,bdh-scat,bdfh}).  In \cite{bfp}, a rigorous  mathematical justification of   Bohr’s frequency condition  was derived
using  an expansion of the scattering amplitudes with respect to powers of the fine structure constant for the Pauli-Fierz model.  In \cite{bkz}, the photoelectric effect has been studied for a model of an atom with a single bound state,  coupled to the quantized electromagnetic  field. A related problem is studying the time-evolution in models of quantum field theory. In \cite{bmw}, this question has been addressed for the Spin-Boson model.
A good overview has been given in \cite{spohn_dynamics_2008}.

This work heavily relies on the multiscale analysis carried out in \cite{bdh-res} as well as  on the results in \cite{bdh-scat}.  We de not repeat any of those proofs here but rather focus on the core argument to derive the above mentioned non-perturbative formula. However, throughout this work, we give precise references to any of the utilized theorems and lemmas which also contain all technical details.

\subsection{The Spin-Boson model}
\label{sec:defmodel}
In  this section  we introduce the considered model and
 state preliminary definitions, well-known tools and facts from which we
start our analysis. If the reader is already familiar with the introductory
Sections 1.1 until 1.3 of \cite{bdh-res}, these subsections can be skipped.

The non-interacting Spin-Boson Hamiltonian is defined as
\begin{align}
\label{h0def}
H_0:=K + H_f , \qquad K:= \begin{pmatrix}
e_1 & 0 \\
0 & e_0
\end{pmatrix} ,
\qquad
H_f:=\int \mathrm{d^3}k \, \omega(k) a(k)^* a(k).
\end{align}
We regard $K$  as an idealized free Hamiltonian of a two-level atom. As already
stated in the introduction, its two energy levels are denoted by the real
numbers $0 = e_0 <e_1$. Furthermore, $H_f$ denotes the free
Hamiltonian of a massless scalar field having dispersion relation
$\omega(k)=|k|$, and $a,a^*$ are the annihilation and creation operators on the
standard Fock space.
For a precise defintion we refer to \cite[Section 1.1]{bdh-scat}.
Below, we
sometimes call $K$ the atomic part, and $H_f$ the free field part of the
Hamiltonian.  The sum of the free two-level atom Hamiltonian $K$ and the free
field Hamiltonian $H_f$   is named ``free Hamiltonian''
$H_0$.  The interaction term reads
\begin{align}
\label{interaction}
V:= \sigma_1\otimes \left( a(f) + a(f)^*\right) , \qquad   \sigma_1:= \begin{pmatrix}
0 & 1 \\
1 & 0
\end{pmatrix} ,
\end{align}
where the boson form factor  is given by
\begin{align}
f: \R^3 \setminus \{0\}\to \R , \qquad k\mapsto e^{-\frac{k^2}{\Lambda^2}}|k|^{-\frac{1}{2}+\mu} .
\label{eq:f}
\end{align}
In our case, the gaussian factor in \eqref{eq:f} acts as an ultraviolet cut-off for $\Lambda>0$ being the ultraviolet cut-off parameter and in addition the fixed number
\begin{align}
\label{const:mu}
\mu\in (0,1/2)
\end{align}
 yields a regularization of the infrared singularity at $k=0$ which is a
 technical assumption chosen such that we can apply  the results obtained in  \cite{bdh-res}.
 Note that the form factor $f$ only depends on the
 radial part of $k$. To emphasize this,  we often write $f(k)\equiv f(|k|)$.

The full Spin-Boson Hamiltonian is then defined as
\begin{align}
\label{eq:H}
H:= H_0 + g V
\end{align}
for some  coupling constant $g>0$ on the
Hilbert space
\begin{align}
\mathcal H := \mathcal K \otimes \mathcal F\left[ \mathfrak{h}\right] , \qquad
\mathcal K:= \C^2,
\end{align}
where
\begin{align}
\mathcal F\left[ \mathfrak{h}\right] :=  \bigoplus^\infty_{n=0} \mathcal F_n\left[ \mathfrak{h}\right]
,\qquad
\mathcal F_n\left[ \mathfrak{h}\right] :=
\mathfrak{h}^{\odot n},\qquad
\mathfrak h:= L^2(\mathbb R^3,\C)
\end{align}
denotes the standard bosonic Fock space, and superscript $\odot n$ denotes the
n-th symmetric tensor product,  where by convention $\mathfrak{h}^{\odot 0}\equiv
\C$. Note that we identify $K\equiv K\otimes 1_{\mathcal F[\mathfrak h]}$ and
$H_f\equiv 1_{\mathcal K}\otimes H_f$ in our notation (see Notation \ref{R} below).

An element $\Psi \in
\mathcal F[\mathfrak{h}]$
can be represented as a sequence $(\psi^{ (n)})_{n\in\N_0}$ of wave functions $\psi^{ (n)} \in \mathfrak{h}^{\odot n}$. The state $\Psi$ with $\psi^{ (0)}=1$ and $\psi^{ (n)}=0$ for all $n\geq 1$ is called the vacuum and is denoted by
\begin{align}
\label{Omega}
\Omega:=(1,0,0,\dots)\in \mathcal F\left[ \mathfrak{h}\right] .
\end{align}
Note that $a$ and $a^*$ fulfill the canonical commutation
relations:
\begin{align}
\label{eq:ccr}
  \forall h,l\in\mathfrak{h}, \qquad \left[a(h),a^*(l)   \right]=\left\langle h, l\right\rangle_2 , \qquad \left[a (h),a(l)   \right]=0 , \qquad \left[a^*(h),a^*(l)   \right]=0.
\end{align}
Let us recall some well-known facts about the introduced model.
It is well-known that $K,H_f,H_0,H$ are self-adjoint and bounded below on the domains $\mathcal K,\mathcal D(H_f),\mathcal D(H_0),\mathcal D(H)$, respectively (see, e.g., \cite[Proposition 1.1]{bdh-scat}).  The spectrum of $K$ consists of two eigenvalues $e_0$ and $e_1$ and the corresponding eigenvectors are
\begin{align}
\label{varphi}
\varphi_0= \left( 0,1 \right)^T \qquad \text{and} \qquad  \varphi_1= \left(1,0\right)^T \qquad \text{with} \qquad K \varphi_i =e_i \varphi_i , \quad i=0,1.
\end{align}
The spectrum of $H_f$ is  $\sigma (H_f)= [0, \infty )$ and it
 is absolutely continuous (see \cite{reedsimon2}). Consequently, the spectrum of $H_0$ is given by
$\sigma (H_0)=  [e_0, \infty )$, and $e_0,e_1$ are eigenvalues embedded in the absolutely continuous part of the spectrum of $H_0$ (see \cite{reedsimon1}).
\begin{notation}
\label{R}
    In  this  work we omit
    spelling out identity operators whenever unambiguous.
    For every vector spaces $V_1$,  $V_2$ and
    operators $ A_1 $ and $A_2$ defined on $V_1$ and $V_2$, respectively, we
    identify \begin{equation}\label{iden} A_1 \equiv A_1 \otimes \mathbbm
        1_{V_2}, \hspace{2cm}  A_2  \equiv \mathbbm 1_{V_1} \otimes A_2 .
    \end{equation}
    In order to simplify our notation further, and whenever
    unambiguous, we do not utilize specific  notations for every inner product
    or norm that we  employ.
\end{notation}

\subsection{Complex dilation}
\label{sec:dil}
In this section we shortly introduce the method of complex dilation which is a key tool for proving our main result. For a more detailed presentation we refer to \cite[Section 1.2]{bdh-scat}. We start by defining a family of
unitary operators on $\mathcal H$ indexed by $\theta \in\R$.
\begin{definition}
\label{def:dil}
For $\theta \in \mathbb R$, we define the unitary
    transformation
\begin{align}
u_\theta: \mathfrak{h}&\to \mathfrak{h}
, \qquad \psi(k) \mapsto e^{-\frac{3\theta}{2}} \psi(e^{-\theta}k) .
\end{align}
Similarly, we define its canonical lift $U_\theta: \mathcal F [\mathfrak{h}]\to
\mathcal F [\mathfrak{h}]$ by the lift condition $U_\theta a(h)^*
U_\theta^{-1}=a(u_\theta h)^*$, $h\in\mathfrak{h}$,   and $U_\theta \Omega=\Omega$.   This defines $U_\theta$
uniquely. With slight abuse of
notation, we also denote $\mathbbm 1_{\mathcal K}\otimes U_\theta$ on $\mathcal
H$ by the same symbol $U_\theta$.

We say that
$ \Psi \in \mathcal F[\mathfrak{h}]$ is an analytic vector  if the map $ \theta
 \mapsto \Psi^\theta := U_\theta \Psi $  has an analytic continuation  from an open connected set in the real line to a (connected) domain in the complex plane.
\end{definition}
We define the family of transformed Hamiltonians, for $\theta \in
 \R$,
\begin{align}\label{Hthetaaaa}
H^\theta :=U_\theta H U_\theta^{-1} =K + H^\theta_f +g V^\theta,
\end{align}
where
\begin{align}
    \label{eq:Hftheta}
H_f^\theta:= \int \mathrm{d^3}k \, \omega^\theta(k) a^*(k) a(k) , \qquad
V^\theta:= \sigma_1 \otimes  \left(a(f^{\overline \theta})+
a(f^{\theta})^*  \right)
\end{align}
and
\begin{align}
\label{def:thetafncts}
\omega^\theta(k):= e^{-\theta}|k|, \qquad f^\theta: \R^3\setminus \{0\}\to\R , \quad k\mapsto e^{-\theta (1+\mu)} e^{-e^{2\theta}\frac{k^2}{\Lambda^2}}|k|^{-\frac{1}{2}+\mu}.
\end{align}
  Eqs.\ \eqref{def:thetafncts}, \eqref{eq:Hftheta} and the
right-hand side of  \eqref{Hthetaaaa} can be defined for complex
$\theta$ (see, e.g., \cite[Lemma 1.4]{bdh-scat}).
For sufficiently small coupling constants and
 $\theta \in \mathcal S$, where $\mathcal S$ is a suitable subset of the complex plane defined in   \eqref{def:setS}  below,
it has been shown that $H^\theta$ has two non-degenerate eigenvalues
$\lambda^\theta_0$ and $\lambda^\theta_1$ with corresponding rank one
projectors denoted by $P^\theta_0$ and $P^\theta_1$, respectively; see, e.g.,
\cite[Proposition 2.1]{bdh-res}.
Note that there the $\theta$-dependence was omitted in the notation. For convenience of the reader, we make it explicit in this paper.
 The corresponding dilated eigenstates can,
therefore, be written as
\begin{align}
\label{eq:gsvec}
\Psi^\theta_{\lambda_i}:=  P_i^\theta  \varphi_i\otimes \Omega  , \qquad i=0,1 .
\end{align}
where the eigenstates $\varphi_i$ of the free atomic system are given in
\eqref{varphi}, and $\Omega$ is the bosonic vacuum defined in \eqref{Omega}.
In our notation $\Psi^\theta_{\lambda_i}$ is not necessarily normalized.  We
know from \cite[Theorem 2.3]{bdh-res} that the eigenvalues $\lambda^\theta_i$
are independent of $\theta$
as long as  $\theta $ belongs to   $\mathcal S$
 and, therefore, we suppress it in our notation
writing $\lambda^\theta_i\equiv\lambda_i$. In the case that $i = 1$, it is necessary that $0$  does not belong to  $  \mathcal S  $.
This is not required  if $i=0$, and in this situation we extend the set $\mathcal{S}$, with the same notation, to an open connected set that contains $0$  (see  \cite[Definition 1.4 and Remark 2.4]{bdh-res}). From this, it is easy to see that  $\Psi^{\theta=0}_{\lambda_0}= \Psi_{\lambda_0}$ - as introduced above.

\subsection{Scattering theory}
\label{sec:scattering}
Finally, we give a short review of scattering theory which is necessary to
state the main result  in
Section~\ref{sec:mainresult}. For a more detailed introduction we refer to \cite[Section 1.3]{bdh-scat}.

\begin{definition}[Basic components of scattering theory]
\label{defasymptop}

We denote by
    $\mathfrak{h}_0$
 the set of     smooth complex-valued
functions on $\R^3  $ with compact support contained in
 $\R^3  \setminus \{0 \}  $.

We define the following objects:
\begin{enumerate}
    \item[(i)]
    For $h\in\mathfrak{h}_0$ and $\Psi \in \mathcal K\otimes \mathcal D(H_f^{1/2})$, the asymptotic annihilation operators
\begin{align}
    \label{asymptop}
    a_\pm(h)\Psi :=\lim\limits_{t\to\pm \infty}a_t(h)\Psi, \quad
    a_t(h):=e^{itH}a(h_t) e^{-itH},
    \quad
  h_t(k):=h(k) e^{ - it\omega(k)} .
\end{align}
Moreover, we define the asymptotic creation operators
 $a_\pm^*(h)$ as the respective adjoints.
\item[(ii)] The asymptotic Hilbert spaces
\begin{align}
\mathcal H^\pm :=\mathcal K^\pm \otimes \mathcal F\left[\mathfrak{h}\right]
\quad \text{where} \quad  \mathcal K^\pm:=\left\lbrace \Psi\in \mathcal H :
a_\pm(h) \Psi=0 \,\,\,\, \forall h\in \mathfrak{h}_0    \right\rbrace .
\end{align}
\item[(iii)] The wave operators
\begin{align}
\label{intertwining}
&\Omega_\pm :\mathcal H^\pm \to \mathcal H
\\ \notag
&\Omega_\pm \Psi \otimes a^*(h_1)...a^*(h_n) \Omega:=a^*_\pm (h_1)...a^*_\pm (h_n) \Psi, \quad h_1,...,h_n \in \mathfrak{h}_0,  \quad \Psi \in \mathcal K^\pm .
\end{align}
\item[(iv)] The scattering operator
 $S:=\Omega^*_+\Omega_-$.
\end{enumerate}
\end{definition}
The limit operators $a_\pm$ and $a_\pm^*$ are called asymptotic
outgoing/ingoing annihilation and creation operators. The existence of the
limits in \eqref{asymptop}  and their
properties  (for example  that $\Psi_{\lambda_0}\in\mathcal K^\pm$) are well-known (see e.g.
\cite{fau1,fau2,fau3,fgs1,fgs2,rgk,rk,rk2,derezinski,bkz}). For a detailed proof we refer to \cite[Lemma 4.1]{bdh-scat}.
We can thus define the following  scattering matrix coefficients for one-boson processes:
\begin{align}
\label{eq:2bodyscat}
S(h,l)= \norm{\Psi_{\lambda_0}}^{-2}\left\langle
a^*_+(h)\Psi_{\lambda_0},a^*_-(l) \Psi_{\lambda_0} \right\rangle, \qquad \forall
h,l\in \mathfrak h_0 ,
\end{align}
where the factor $\norm{\Psi_{\lambda_0}}^{-2}$ appears due to the fact that,
as already mentioned above,  in our notation, the ground state
$\Psi_{\lambda_0}$ is not necessarily normalized.  In addition, it will be
convenient to work with the corresponding  transition matrix
coefficients for one-boson processes given by
\begin{align}
T(h,l)=S(h,l) -\left\langle h , l \right\rangle_2 \qquad \forall h,l\in \mathfrak h_0 .
\label{eq:Tmatrix}
\end{align}
Physically, these matrix coefficients may be interpreted as transition
amplitudes of the scattering process in which an incoming boson with wave
function $l$ is scattered at the two-level atom into an outgoing boson with
wave function $h$.  Notice that the transition matrix coefficients of
multi-boson processes can be defined likewise but in this work we focus on
one-boson processes only.

\section{Main results}
\label{sec:mainresult}

We are now able to state our main result which  provides
the precise relation between the one-boson transition matrix elements and the  resolvent of the complex
dilated Hamiltonian. The corresponding proofs will be
provided in Section \ref{sec:proof-mainresult}.
\begin{theorem}[Scattering Formula]
\label{FK}
For sufficiently small $g$, $ \theta $ in a suitable subset
    $\mathcal{S} \subset \C$ (see \eqref{def:setS}),
and
 for all $h,l\in\mathfrak{h}_0$, the
 transition matrix coefficients for one-boson processes are given by
\begin{align}
\label{scatteringformula}
T(h,l)=\int \mathrm{d^3} k  \mathrm{d^3} k' \, \overline{h(k)} l(k') \delta(\omega(k)-\omega(k')) T(k,k')
\end{align}
where
\begin{align}
\label{scatteringkernel}
T(k,k')
=- 2  \pi ig^2 f(k) f(k')\norm{\Psi_{\lambda_0}}^{-2}\bigg( &
\left\langle    \sigma_1 \Psi^{\overline \theta}_{\lambda_0},\left( H^{ \theta}-\lambda_0-|k'| \right)^{-1} \sigma_1 \Psi^{ \theta}_{\lambda_0}\right\rangle
\notag \\
& +\left\langle    \sigma_1 \Psi^{ \theta}_{\lambda_0},\left( H^{\overline \theta}-\lambda_0+|k'| \right)^{-1} \sigma_1 \Psi^{\overline \theta}_{\lambda_0}\right\rangle
 \bigg)  .
\end{align}
\end{theorem}
The integral with respect to the  Dirac's delta distribution
distribution $\delta$  in \eqref{scatteringformula} is
to be understood as
\begin{align}
T(h,l)=\int_0^\infty \mathrm{d}|k| \int \mathrm{d}\Sigma \mathrm{d}\Sigma' \,
\overline{h(|k|, \Sigma)} l(|k|,\Sigma') T(|k|,\Sigma,|k|,\Sigma'),
\end{align}
where we have introduced spherical coordinates $k=(|k|,\Sigma)$ with $\Sigma$
being the solid angle and $T(k,k')\equiv T(|k|,\Sigma,|k|,\Sigma')$ is given by
\eqref{scatteringkernel}.  Notice that \eqref{scatteringkernel} is not defined
for $k=0$ or $k'=0$. However, since we take $h,l\in  \mathfrak{h}_0$, the expression \eqref{scatteringformula}
is well-defined.  Representing such matrix elements in terms of a distribution
kernel is convenient (in our case, e.g., it makes the energy conservation
apparent) and also frequently used in the literature. In particular, similar
distribution kernels in a closely related model have been studied in
\cite{bkz,bdh-scat}.
\begin{remark}
\label{rem:connection}
In a similar vein as in \cite{bdh-scat}, we can apply perturbation theory together with the spectral properties obtained in \cite{bdh-res} in order to deduce a  result as  \cite[Theorem 2.2]{bdh-scat} from Theorem \ref{FK} above. Then, one can again see the Lorentzian shape of the integral kernel which was explained in detail in \cite{bdh-scat}.
\end{remark}
    In the remainder of this  work we denote by $ C $
    any generic (indeterminate), positive constant that might change from line to line but does not depend on the coupling constant.

\section{Proof of the main result}
\label{sec:proof-mainresult}

In the remainder of this work we provide the proof of  Theorem~\ref{FK}. This section has three parts: In
Section~\ref{sec:proof-prelim}, we recall a preliminary formula for the
scattering matrix coefficients; c.f.\ Theorem \ref{intker} below, which was proven in \cite[Theorem 4.3]{bdh-scat}. This formula
together with several technical ingredients provided in
Section~\ref{sec:proof-techingredients} and \ref{sec:key}  pave the way for the proof of
our main result given in Section~\ref{sec:proof-mainthm}.

\subsection{Preliminary scattering formula}
\label{sec:proof-prelim}
The following theorem has been proven in \cite[Theorem 4.3]{bdh-scat}.
\begin{theorem}[Preliminary Scattering Formula]
   \label{intker}
   For $h,l\in \mathfrak{h}_0$, the  transition matrix coefficient for one-boson processes
   $T(h,l)$  defined in \eqref{eq:Tmatrix} fulfills
\begin{align}
    T(h,l)=  \lim\limits_{t\to - \infty}\int \mathrm{d^3}k \mathrm{d^3}k' \,
\overline{h(k)} l(k') \delta(\omega(k)-\omega(k')) T_t(k,k')
\label{T}
\end{align}
for the integral kernel
\begin{align}
\label{intkernel}
T_t(k,k')=-2\pi i g f(k)  \norm{\Psi_{\lambda_0}}^{-2}{\langle \sigma_1 \Psi_{\lambda_0},
a_t(k')^* \Psi_{\lambda_0}\rangle}.
\end{align}
\end{theorem}
The integral in \eqref{T}
is to be understood as
\begin{align}
T(h,l)= -2\pi ig \norm{\Psi_{\lambda_0}}^{-2}\bigg\langle
 \sigma_1  \Psi_{\lambda_0}, a_-(W)^* \Psi_{\lambda_0} \bigg\rangle
 \label{eq:Tprecise1}
\end{align}
for $W\in\mathfrak{h}_0$ given by
\begin{align}
\label{def:W1st}
\R^3\ni k\mapsto W( k):=|k|^2 l(k) \int\mathrm{d}\Sigma \, \overline{h(|k|,\Sigma)}f(|k|,\Sigma)
\end{align}
using spherical coordinates $k=(|k|,\Sigma)$ with $\Sigma$ being the solid
angle.

\subsection{General ingredients for the proof of the main theorem}
\label{sec:proof-techingredients}
Here, we state some general results which are applied in the proof of our main
  result, see Section \ref{sec:proof-mainthm}.  Most of the statements in this
section are   formulated without motivation. However, their importance
 becomes clear later in Section \ref{sec:proof-mainthm}.
At first, we recall a representation formula of the time-evolution operator similar
to the Laplace transform representation (see, e.g., \cite{bach}). This formula is an important ingredient for the proof of the perturbative scattering formula in \cite{bdh-scat} and it  plays a relevant role in the present work. For a detailed proof we refer to \cite[Lemma 4.5]{bdh-scat}.
\begin{lemma}
\label{laplace}
For $\epsilon>0$, $\nu =\Im \theta >0$ and sufficiently large $R >0$, we consider the concatenated contour
$\Gamma(\epsilon,R):=\Gamma_{-}(\epsilon,R)\cup
\Gamma_{c}(\epsilon)\cup \Gamma_{d}(R)$ (see Figure \ref{fig:CurveGamma}),
where
\begin{align}
\Gamma_{-}(\epsilon,R)&:=[-R, \lambda_0 -\epsilon]\cup [\lambda_0 +\epsilon,R
],
\notag \\
\Gamma_{d}(R)&:=\left\lbrace -R -ue^{i\frac{\nu}{4}} :u\geq 0 \right\rbrace
\cup \left\lbrace R +ue^{-i\frac{\nu}{4}} :u\geq 0 \right\rbrace ,
\notag \\
\Gamma_{c}(\epsilon)&:=\left\{ \lambda_0 -\epsilon e^{-it}: t\in [0,\pi]
\right\}.
\label{Gamma-parts}
\end{align}
The  orientations of the contours in \eqref{Gamma-parts} are given by the arrows depicted in Figure
\ref{fig:CurveGamma}.
Then, for all analytic vectors $\phi,\psi\in\mathcal H$
(analytic in a --  connected --  domain containing   $0$)
    and $t>0$, the following identity holds true:
\begin{align}
\left\langle \phi , e^{-itH} \psi  \right\rangle =
\frac{1}{2\pi i}\int_{\Gamma(\epsilon,R)}\mathrm{d}z\, e^{-itz} \left\langle \psi^{\overline{\theta}} , \left( H^\theta-z  \right)^{-1} \phi^\theta \right\rangle  .
\label{eq:laplace}
\end{align}
\end{lemma}
\begin{figure}[h]
\centering
\includegraphics[width=\textwidth]{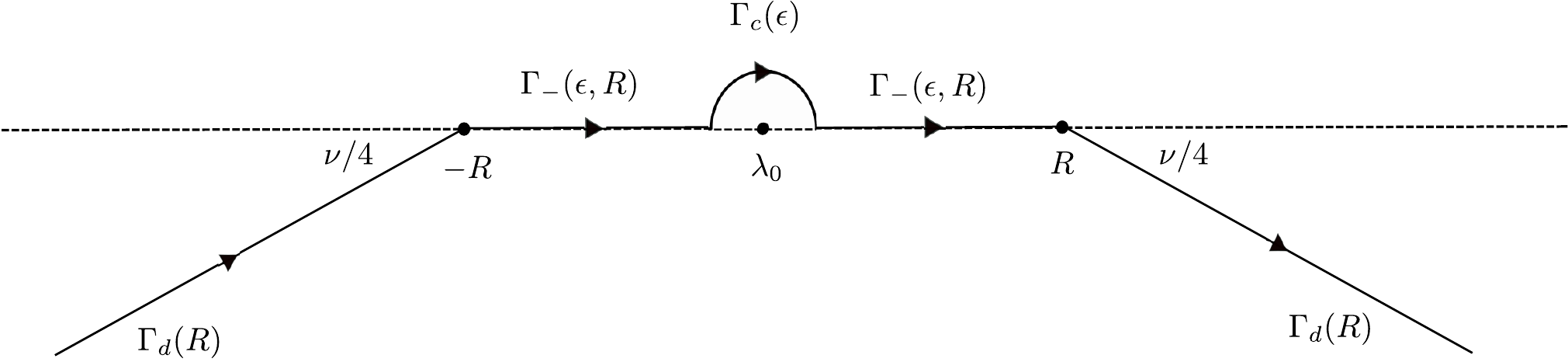}
  \caption{An illustration of the contour
  $\Gamma(\epsilon,R):=\Gamma_{-}(\epsilon,R)\cup \Gamma_{c}(\epsilon)\cup
  \Gamma_{d}(R)$. }
    \label{fig:CurveGamma}
\end{figure}

In this paper we use a non-standard definition of the Fourier transform and its inverse:
\begin{align}\label{def:fourier-distri}
\mathfrak{F}[u](x):= \int_\R \mathrm{d}s \,  u(s)e^{-isx},
\qquad
\mathfrak{F}^{-1}[u](x):= (2
\pi)^{-1}\int_\R \mathrm{d}s \,  u(s)e^{isx},
\end{align}
where $u \in \mathit S(\R,\C)  $ (the Schwartz space).
We utilize the same symbols (and names) for their dual transformation on
$S'(\R,\C)$ (the space of tempered distributions). We identify, as usual,  functions $f\in L^p(\R,\C)$ (for some $p \in [1, \infty]$) with their induced tempered distributions in $  \mathit S'(\R,\C) $ ($f(u) = \int uf  $) and, similarly, we  identify  functions $f\in L^1_\text{loc}(\R,\C)$   with their induced distributions in $   \big( C_0^\infty(\R,\C) \big )' $.  We denote by $\Theta$ the Heaviside function (or distribution, or tempered distribution) and by $\delta$ the Dirac $\delta $ distribution (or tempered distribution):
\begin{align}\label{def:delta-heavi}
 \Theta(x):=\begin{cases} 1 \quad &\text{for}
\quad x\geq 0 \\   0 \quad &\text{for} \quad x< 0 \end{cases},  \hspace{1cm} \Theta(u) = \int_0^{\infty} u(x)dx,  \hspace{1cm} \delta (u) = u(0),
\end{align}
for $u \in  \mathit S'(\R,\C) $.
\begin{lemma}
\label{lemma:heaviside}
We denote by $\left(\text{PV}\left(1/\bullet \right)\right)\in  S'(\R,\C)  $
the principal value:
\begin{align}
\label{def:princval}
\left(\text{PV}\left(1/ \bullet\right)\right)(\varphi) \equiv
\text{PV}\int_\R\mathrm{d}s \, \frac{1}{s} \varphi(s) :=\lim\limits_{\eta\to
0^+} \int_{\R\setminus [-\eta,\eta]}\mathrm{d}s \, \frac{1}{s} \varphi(s)
\qquad \forall \varphi \in  \mathit S(\R,\C)     .
\end{align}
  It follows that
\begin{align}
\mathfrak{F}[\Theta]= \pi \delta -i \text{PV}\left(1/\bullet \right) .
\end{align}
\end{lemma}
The above result can be shown using methods from standard distribution theory.
However, for the sake of completeness, we present a proof in Appendix
\ref{app:heaviside}.

\subsection{Key estimates}
\label{sec:key}

In this section we establish two key estimates for the proof of the main theorem.  We point out to the reader that they   strongly rely on the results obtained in \cite{bdh-res}. However, for simplicity and due to the fact that the important features have already been studied in \cite[Section 4.3]{bdh-scat}, we  omit the details related to the multiscale analysis carried out in \cite{bdh-res}, and give precise references instead.

\begin{definition}{(c.f.\ \cite[Definition 4.6]{bdh-scat})}  \label{Defen} For every fixed  numbers $\rho_0 \in (0,1)$ and  $ \rho \in (0, \min(1, e_1/4) )$  satisfying  \eqref{dorm2}, we define the sequences
\begin{align}
\rho_n := \rho_0 \rho^n,   \hspace{1cm} \epsilon_n :=  20  \rho_n^{1+ \mu/4}, \hspace{1cm} \forall n \in \mathbb{N}.
\end{align}
\end{definition}
\begin{lemma}
\label{lemma:thl12}
Set  $G\in \mathcal
C^\infty_c(\R \setminus\{0\},\C)$, $ n \in \N$ large enough and $\eta > 0$ small enough  such that  $G(x )= 0,$ for $  |x| \leq 2 ( \epsilon_n + \eta) $.   We define
\begin{align}
\label{eq:t_eps,R,eta}
 T_{n,R}(\eta):&= \int_{\Gamma_-(\epsilon_n,R)}  \mathrm{d}z\,   u(z)
  \int_{\R}\mathrm{d}r \, \frac{G(r)}{z-\lambda_0-r}\left( 1-\mathbbm1_{I_\eta(z)}(r) \right) ,
\end{align}
where $  \mathbbm1_{I_
\eta(z)} $ is the characteristic function of the set $I_\eta(z):=[z-\lambda_0-\eta,z-\lambda_0+\eta]$,  $\Gamma_-(\epsilon_n,R)$ is defined in \eqref{Gamma-parts}  and
\begin{align}
\label{eq:defu}
u: \overline{\C^+} \setminus \{\lambda_0\}\to\C, \qquad z \mapsto u(z):=\left\langle    \sigma_1 \Psi^{\overline \theta}_{\lambda_0},\left( H^\theta-z  \right)^{-1} \sigma_1 \Psi^\theta_{\lambda_0}\right\rangle  .
\end{align}
Then, for sufficiently large $R$ (independent of $n$ and $\theta \in \mathcal{S}$), there is a constant $C$ (that does not depend on $ n $, but it does depend on  $G$,  $\theta$, $e_1$ and $m$ -- see above \eqref{eq:defcone} below) such that
\begin{align}
\label{eq:t_eps,R,eta_lim}
\left| T_{n,R }(\eta) -  \pi i\int_{\R}\mathrm{d}r \, G(r)u(r+\lambda_0) \right|\leq
 C \left( \rho_n^{\mu/8} +\frac{1}{R}  +  \eta \right).
\end{align}
\end{lemma}
\begin{proof}
The integrand in \eqref{eq:t_eps,R,eta} is absolutely integrable with respect to the variables $z$ and $r$ because the singularity is cut off by  the characteristic function. We apply Fubini's theorem to get
\begin{align}\label{chucha1}
 T_{n,R} (\eta) &= \int_{\R}\mathrm{d}r \,  G(r) \int_{\Gamma_-(\epsilon_n ,R)}  \mathrm{d}z\,   u(z)
   \, \frac{1}{z-\lambda_0-r}\left( 1-\mathbbm1_{  I_\eta(z)}(r) \right) .
\end{align}
Next, we analyze the inner integral above for $r$ in the support of $G$.
Set $\Gamma_{(r)}$ the half circle in the upper half complex plane with
center $r + \lambda_0$ and radius $ \eta$.
Moreover, set $ \Gamma^{(R)} $ the  half-circle in the upper half complex
plane with center $0 $ and radius $R$.
\begin{figure}[h]
\centering
\includegraphics[width=0.6\textwidth]{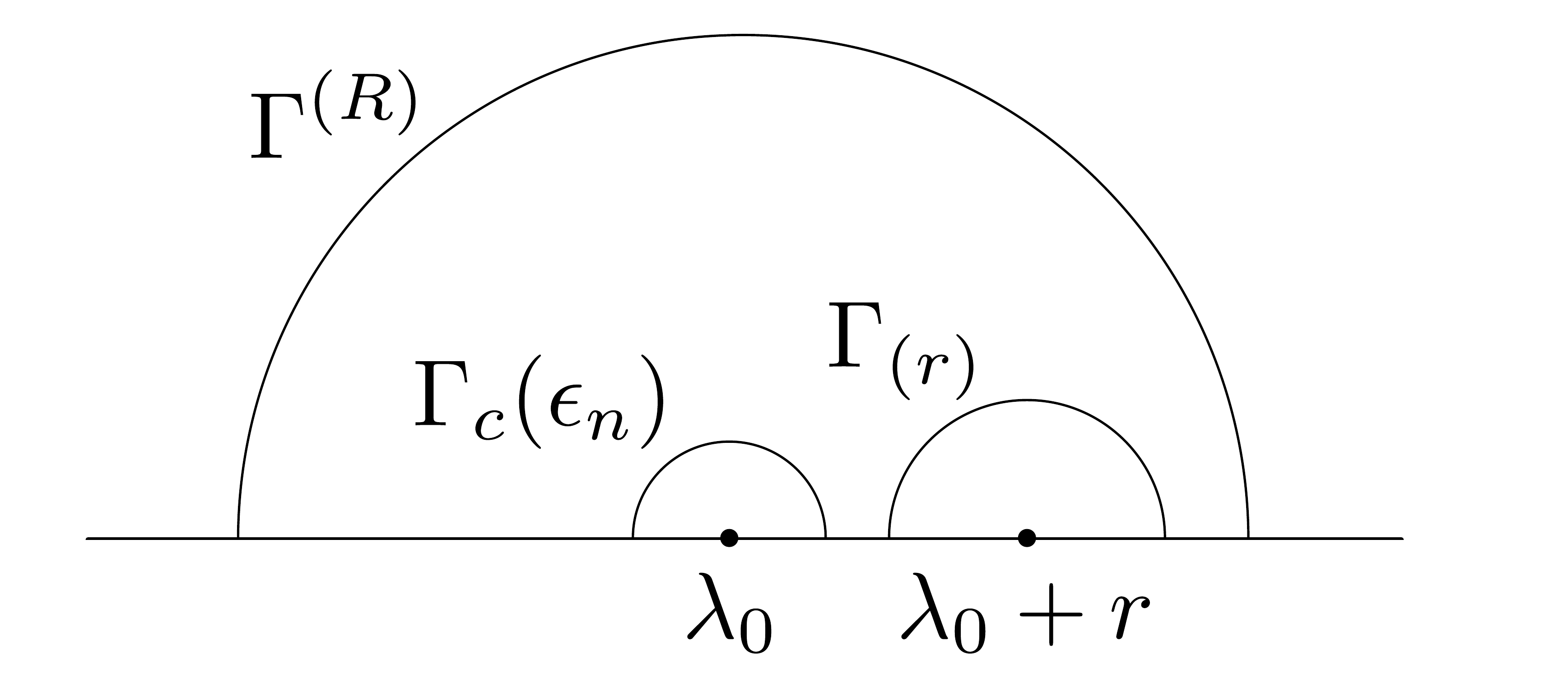}
  \caption{An illustration of the half circles $\Gamma_c(\epsilon_n)$ and $\Gamma_{(r)}$.}
    \label{fig:half-circles}
\end{figure}
As despicted in Figure \ref{fig:half-circles},
the two half circles $\Gamma_c(\epsilon_n)$ and $\Gamma_{(r)}$ do not
intersect each other for all $r$ in the support of $G$. This is a consequence of the assumption that the support of $G$ does not intersect with the interval
$(-2(\epsilon_n +\eta), 2(\epsilon_n +\eta))$. Moreover, we find that  both half circles $\Gamma_c(\epsilon_n)$ and $\Gamma_{(r)}$ are contained in $  \Gamma^{(R)}$ for large enough $R$
(the value of $R$ can be chosen uniformly with respect to $n$ and $\theta \in \mathcal{S}$, but it depends on the support of $G$  independent of $ n $ and $\theta \in \mathcal{S}$, but dependent on the support of $G$).

Note that there is a constant $C$ (that depends on the support of $G$, but not on $n$, $\theta \in S$, $\rho$ and $\rho_0$) such that (see \eqref{resA})
\begin{align}\label{chucha2}
\Big | u(z)  \, \frac{1}{z-\lambda_0-r} \Big | \leq \frac{C}{R^2}, \qquad \forall z \in   \Gamma^{(R)} .
\end{align}
Moreover,  there is a constant $ C $ (that depends on the support of $G$, but not on $n$,  $\rho$ and $\rho_0$) such that (see \eqref{nn0})
\begin{align}\label{chucha3}
\Big | u(z)  \, \frac{1}{z-\lambda_0-r} \Big | \leq C \boldsymbol{C}^{n+1} \frac{1}{\rho_n}, \hspace{2cm} \forall z \in   \Gamma_c(\epsilon_n),
\end{align}
where   $\rho_n=\rho_0\rho^n$ and $\rho_0>0$, $0<\rho<1$ and $\boldsymbol C>0$ are specific numbers defined in \cite[Definition 4.1 and 4.2]{bdh-res} and fulfilling \eqref{dorm2}.
We know from \eqref{spectrum} and \eqref{eq:impartres} that the only spectral point of $H^{\theta}$ in $\overline{\C^+}$ is $\lambda_0$. Hence, there is a constant  $C$ (that depends on the support of $G$, but not on $n$) such that
\begin{align}\label{chucha4}
|   u(z) - u( \lambda_0 + r  )  | \leq C \eta,   \qquad \forall z \in   \Gamma_{(r)} .
\end{align}
A direct calculation shows that
\begin{align}\label{chucha5}
\int_{\Gamma_{(r)}} dz \,  u(   \lambda_0 + r )  \, \frac{1}{z-\lambda_0-r} =
- u(   \lambda_0 + r )   i\pi .
\end{align}
We choose the contour which follows the following set of points
$ \Big ( \Gamma_-(\epsilon_n,R) $     $\setminus (r + \lambda_0 - \eta, r + \lambda_0 + \eta )  \Big ) $  $   \cup \Gamma^{(R)} $   $ \cup \Gamma(r) \cup  \Gamma_c(\epsilon_n)$
along the mathematical positive orientation. This is a closed contour where the function $z \mapsto  \frac{u(z)}{z-\lambda_0-r}  $ is continuous, and an it is analytic on its interior.
 Then, it follows from Cauchy's integral formula that (notice that, for $z$ in the real numbers,  $  \mathbbm1_{  I_\eta(z)  }(r)  = \mathbbm1_{ [r + \lambda_0 - \eta, r + \lambda_0 + \eta ]  }(z) $)
\begin{align}\label{chucha6}
\int_{\Gamma_-(\epsilon_n,R)}  \mathrm{d}z\,  \frac{u(z)}{z-\lambda_0-r}\left( 1-\mathbbm1_{  I_\eta(z)  }(r) \right)
=   &  \int_{\Gamma_-(\epsilon_n,R)}  \mathrm{d}z\, \frac{u(z)}{z-\lambda_0-r}\left( 1-    \mathbbm1_{ [r + \lambda_0 - \eta, r + \lambda_0 + \eta ]  }(z)  \right)  \notag
\\ =  &  \int_{\Gamma_-(\epsilon_n,R)  \setminus (r + \lambda_0 - \eta, r + \lambda_0 + \eta ) }  \mathrm{d}z\, \frac{u(z)}{z-\lambda_0-r}  \notag
 \\    = & -
   \int_{ \Gamma^{(R)} \cup \Gamma(r) \cup  \Gamma_c(\epsilon_n)   } dz\,
   \frac{u(z)}{z-\lambda_0-r},
\end{align}
which together with \eqref{chucha1}-\eqref{chucha5} imply the desired result, we additionally use  Definition \ref{Defen}  and \eqref{dorm2}   to estimate the integral over $  \Gamma_c(\epsilon_n)   $.
\end{proof}
\begin{lemma}
\label{lemma:uniform}
Let $n \geq 2$ and  $R>0$ be large enough.  For  $0<q<1<Q< \infty$ and $\zeta \in \mathit S(\R,\C)$, we define
\begin{align}
A(Q,n,R):=\int^{Q}_{q} \mathrm{d}s\,  \zeta(s) \int_{\Gamma_-(\epsilon_n,R)}  \mathrm{d}z\, e^{-is(z-\lambda_0)}
    \left\langle    \sigma_1 \Psi^{\overline \theta}_{\lambda_0},\left( H^\theta-z
    \right)^{-1} \sigma_1 \Psi^\theta_{\lambda_0}\right\rangle .
\end{align}
Then, the  limits $A(Q,\infty,\infty):=\lim\limits_{n,R\to\infty}A(Q,n,R)$ and $A(\infty,n,R):=\lim\limits_{Q\to\infty}A(Q,n,R)$ exist
and they are uniform with respect to $Q$ and $(n, R)$, respectively. Moreover,   there is a constant $C$ (independent of $n$, $q$, $Q$  and $R$)    such that
\begin{align}\label{Misto}
|A(Q,n,R)-A(\infty,n,R)|\leq C/Q.
\end{align}
Additionally, the limits
\begin{align}
\lim\limits_{Q\to\infty}\lim\limits_{n,R\to\infty}A(Q,n,R),        \hspace{1cm}  \lim\limits_{n,R\to\infty} A(\infty,n,R)
\end{align}
exist and they are equal.
\end{lemma}
\begin{proof}
For $0<q<Q<\infty$, $n \in \N $ and $ R\in \R^+$ sufficiently large, we write
\begin{align}
\label{asplit0}
A(Q,n,R)&=A^{(1)}(Q,n,R)+A^{(2)}(Q,n,R) ,
\end{align}
where
\begin{align}
\label{asplit1}
A^{(1)}(Q,n)&:=\int^{Q}_{q} \mathrm{d}s \, \zeta(s)\int_{I_n}  \mathrm{d}z\, e^{-is(z-\lambda_0)}
    \left\langle    \sigma_1 \Psi^{\overline \theta}_{\lambda_0},\left( H^\theta-z
    \right)^{-1} \sigma_1 \Psi^\theta_{\lambda_0}\right\rangle ,
    \\
    \label{asplit2}
 A^{(2)}(Q,R)  &:= \int^{Q}_{q} \mathrm{d}s \,\zeta(s) \int_{I_1 }  \mathrm{d}z\, e^{-is(z-\lambda_0)}
    \left\langle    \sigma_1 \Psi^{\overline \theta}_{\lambda_0},\left( H^\theta-z
    \right)^{-1} \sigma_1 \Psi^\theta_{\lambda_0}\right\rangle .
\end{align}
 Here, we split the  the domain of integration $\Gamma_-(\epsilon_n,R)=I_1 \cup I_n$, where $I_1:= [-R,R]\setminus (\lambda_0-\epsilon_1, \lambda_0 +\epsilon_1)$ and $I_n:= [\lambda_0-\epsilon_1, \lambda_0 +\epsilon_1]\setminus (\lambda_0-\epsilon_n, \lambda_0 +\epsilon_n)$. We analyze first \eqref{asplit2}.
 We  obtain from the integration by parts formula (in the variable $s$) together with $e^{-is(z-\lambda_0)}=i(z-\lambda_0)^{-1}\partial_s e^{-is(z-\lambda_0)}$ that there is a constant $C$ such that, for $\tilde Q>Q$,
\begin{align}
&A^{(2)}(\tilde Q,R) - A^{(2)}(Q,R)
\notag \\
&=i \int^{\tilde Q}_{Q} \mathrm{d}s \,\zeta(s) \int_{I_1}  \mathrm{d}z\, (z-\lambda_0)^{-1}\partial_s e^{-is(z-\lambda_0)}
    \left\langle    \sigma_1 \Psi^{\overline \theta}_{\lambda_0},\left( H^\theta-z
    \right)^{-1} \sigma_1 \Psi^\theta_{\lambda_0}\right\rangle .
\notag \\
&= i\int_{I_1}   \mathrm{d}z\, \left( \zeta(\tilde Q)e^{-i\tilde Q(z-\lambda_0)}- \zeta(Q)e^{-iQ (z-\lambda_0)}\right)(z-\lambda_0)^{-1}
    \left\langle    \sigma_1 \Psi^{\overline \theta}_{\lambda_0},\left( H^\theta-z
    \right)^{-1} \sigma_1 \Psi^\theta_{\lambda_0}\right\rangle
    \notag \\
&-i\int^{\tilde Q}_{Q} \mathrm{d}s \,\left( \partial_s \zeta(s)\right)
\int_{I_1}  \mathrm{d}z\, (z-\lambda_0)^{-1} e^{-is(z-\lambda_0)}
    \left\langle    \sigma_1 \Psi^{\overline \theta}_{\lambda_0},\left( H^\theta-z
    \right)^{-1} \sigma_1 \Psi^\theta_{\lambda_0}\right\rangle .
\end{align}
Since  $\zeta\in \mathit S(\R,\C)$, there is a constant $C$ such that, for all $s\in \R$,
$|\zeta(s)|,|\partial_s \zeta(s)|\leq C/(1+s^2)$, and hence, there is a constant $C$ such that
\begin{align}
\label{A2'}
\left|A^{(2)}(\tilde Q,R) - A^{(2)}(Q,R)  \right| &\leq  C Q^{-1}
\int_{I_1}  \mathrm{d}z\, \left| z-\lambda_0  \right|^{-1}
\left|    \left\langle    \sigma_1 \Psi^{\overline \theta}_{\lambda_0},\left( H^\theta-z
    \right)^{-1} \sigma_1 \Psi^\theta_{\lambda_0}\right\rangle\right| .
\end{align}
It follows from \eqref{resA} and \eqref{nn0} that there is a constant $C$   (independent of $n$, $R$, $q$ and $Q$) such that
\begin{align}
\label{A2}
\left| A^{(2)}(\tilde Q,R) - A^{(2)}(Q,R)  \right| &\leq  C /Q.
\end{align}
Similarly,  using that $\zeta\in \mathit S(\R,\C)$,
 we find  a constant $C$  (independent of $n$, $R$, $q$ and $Q$) such that
\begin{align}
&\left| A^{(1)}(\tilde Q,n) - A^{(1)}(Q,n)  \right| \leq CQ^{-1} \int_{I_n}  \mathrm{d}z\, \left|    \left\langle    \sigma_1 \Psi^{\overline \theta}_{\lambda_0},\left( H^\theta-z
    \right)^{-1} \sigma_1 \Psi^\theta_{\lambda_0}\right\rangle\right|
    \notag \\
    &\leq C Q^{-1}\sum^{n-1}_{j=1}\int_{I_{j,j+1}}  \mathrm{d}z\, \left|    \left\langle    \sigma_1 \Psi^{\overline \theta}_{\lambda_0},\left( H^\theta-z
    \right)^{-1} \sigma_1 \Psi^\theta_{\lambda_0}\right\rangle\right| ,
\end{align}
where  $I_{j,j+1}:= [\lambda_0-\epsilon_j, \lambda_0 +\epsilon_j]\setminus (\lambda_0-\epsilon_{j+1}, \lambda_0 +\epsilon_{j+1})$.
 We observe from \eqref{nn0} together with Definition \ref{Defen}  that there is a constant $C$  (independent of $n$, $R$, $q$ and $Q$) such that
\begin{align}
\label{A1}
\left| A^{(1)}(\tilde Q,n) - A^{(1)}(Q,n)  \right|
    &\leq CQ^{-1}\sum^{\infty}_{j=1   }\int_{I_{j,j+1}}  \mathrm{d}z\, \frac{\boldsymbol C^{j+2} }{\rho_{j+1}}
    \leq C Q^{-1}\sum^{\infty}_{j=1   }\frac{\boldsymbol C^{j+2} \epsilon_j}{\rho_{j+1}}
    .
\end{align}
From Definition \ref{Defen}  and \eqref{dorm2}, we obtain that
\begin{align}
\left| A^{(1)}(\tilde Q,n) - A^{(1)}(Q,n)  \right| \leq C/Q .
\end{align}
This together with \eqref{A2} implies that there is a constant $C$ such that
\begin{align}
\left| A(\tilde Q,n,R) - A(Q,n,R)  \right| \leq C/Q .
\end{align}
Consequently,       the limit  $ \lim_{\tilde Q \to \infty}    A(\tilde Q,n,R) $ exists and it converges    uniformly with respect to $ n $ and $R$.      We denote the limit by $A(\infty,n,R)=\lim\limits_{Q\to\infty}A(Q,n,R)$.     It follows that \eqref{Misto} holds true.

For fixed $Q$ and $\tilde n >n $ and $\tilde R>R$, we have
\begin{align}
\label{exlim0}
\left| A( Q,\tilde n, \tilde R) - A(Q,n,R)  \right| \leq \left| A( Q,\tilde n, \tilde R) - A(Q,\tilde n,R)  \right| + \left| A( Q,\tilde n,  R) - A(Q,n,R)  \right|
 .
\end{align}
For $\tilde n  $ and $\tilde R$ large enough, employing a similar calculation as in \eqref{A2'}, we get from
 \eqref{asplit0}, \eqref{asplit1}, \eqref{asplit2}   that  there is a constant $C$ (that does not depend on $Q$)   such that
\begin{align}
\label{exlim1}
&\left| A( Q,\tilde n, \tilde R) - A(Q,\tilde n,R)  \right| =  \left| A^{(2)}( Q, \tilde R) - A^{(2)}(Q, R)  \right|
\notag \\
&\leq C' \int_{[-\tilde R,-R]\cup [ R,\tilde R]}  \mathrm{d}z\, \left| z-\lambda_0  \right|^{-1}
\left|    \left\langle    \sigma_1 \Psi^{\overline \theta}_{\lambda_0},\left( H^\theta-z
    \right)^{-1} \sigma_1 \Psi^\theta_{\lambda_0}\right\rangle\right| \leq C /R ,
\end{align}
and furthermore, similarly as in \eqref{A1}, we obtain that there is a constant $C$ such that
\begin{align}
 \left| A( Q,\tilde n,  R) - A(Q,n,R)  \right|  &=  \left| A^{(1)}( Q,\tilde n) - A^{(1)}(Q,n)  \right|
    \leq C \sum^{\tilde n-1}_{j=n}\frac{\boldsymbol C^{j+2} \epsilon_j}{\rho_{j+1}}
 ,
\end{align}
and consequently, it follows from Definition \ref{Defen} together with \eqref{dorm2}  that there is a constant $C$
   (that does not depend on $Q$)
 such that
\begin{align}
\label{exlim2}
 \left| A( Q,\tilde n,  R) - A(Q,n,R)  \right|  &
    \leq C /n
 .
\end{align}
This together with \eqref{exlim0} and \eqref{exlim1} yields that there there is a constant $C$    (that does not depend on $Q$)    such that
\begin{align}
\label{ex1000}
\left| A( Q,\tilde n, \tilde R) - A(Q, n,R)  \right|  \leq C(R^{-1}+n^{-1}).
\end{align}
We conclude that
 the limit $A(Q,\infty,\infty):=\lim\limits_{n,R\to\infty}A(Q,n,R)$ exists (uniformly with respect to $Q$). This completes the first part of the lemma.

   Now we prove the second part of the  lemma.
At first, we show   the existence of the limit $\lim\limits_{n,R\to \infty }A(\infty ,n,R)$. For $\tilde n > n$ and $\tilde R>R$, we estimate
\begin{align}
\label{exlim00}
&\left| A(\infty,\tilde n, \tilde R ) - A(\infty ,n,R)  \right|
 \\\notag
&\leq
 \left| A(\infty,\tilde n, \tilde R  ) - A( Q,\tilde n, \tilde R )  \right|
 +\left| A(Q,\tilde n, \tilde R ) - A(Q,n ,R)  \right|
 +\left| A( Q,n, R ) - A(\infty ,n,R)  \right|
 .
\end{align}
For $\epsilon>0$, we take  $Q_0 >0$  such that  for all  $  Q \geq Q_0  $
\begin{align}
\label{epsdrittel}
 \left| A( \infty, \tilde n,\tilde R ) - A( Q,\tilde n ,\tilde R)  \right| \leq \epsilon /3 \quad \text{and} \quad  \left| A( \infty,  n, R ) - A( Q, n , R)  \right| \leq \epsilon /3 .
\end{align}
We obtain from     \eqref{ex1000}    that, for  $\epsilon>0$, there are constants $n_0,R_0 >0$  such that,   for all $n, \tilde n      \geq n_0$ and $R,  \tilde R     \geq R_0$,
\begin{align}
\label{epsdrittel'}
\left| A(Q,\tilde n, \tilde R ) - A(Q,n ,R)  \right| \leq \epsilon/3 .
\end{align}
This together with \eqref{epsdrittel} and \eqref{exlim00} yields that, for  $\epsilon>0$,    there are   $n_0>0$ and $R_0>0$ such that, for  $n\geq n_0$ and $R\geq R_0$, we have
\begin{align}
\label{exlim00'}
&\left| A(\infty,\tilde n, \tilde R ) - A(\infty ,n,R)  \right|   \leq \epsilon .
\end{align}
This implies the existence of the limit $\lim\limits_{n,R\to \infty}A(\infty ,n,R)    =: A(  \infty, \infty, \infty    )       $.
  We fix $ \epsilon  > 0 $.
According to \eqref{exlim00'} we obtain that for large enough $  n, R $,   $    |        A(  \infty, \infty, \infty    )     -      A(  \infty, \ n ,  R    )       | < \epsilon /3  $.
Since  $ \lim_{Q \to \infty } A( Q, n , R) =    A(\infty , n , R)   $ uniformly with respect to $ n, R $, then for large enough $ Q $  (independently of $n, R$)
$    |    A(\infty,  n,   R  )  -A(Q,  n ,  R)  | < \epsilon/3  $. Moreover,  because  $A(Q,\infty,\infty) =\lim\limits_{n,R\to\infty}A(Q,n,R)$ (uniformly with respect to $Q$), for large enough $   n, R $ (independently of $Q $) we have that  $    | A(Q,n ,R) -   A(Q ,\infty,\infty)  |  < \epsilon/3  $.  We conclude that
 there are  $\boldsymbol{n} \in \mathbb{N}$, $  \boldsymbol{R}   >0$ and $ \boldsymbol{ Q } >0$ such that, for  $n\geq \boldsymbol{ n}$, $Q \geq \boldsymbol{ Q } $  and  $R\geq  \boldsymbol {R }$, we have
\begin{align}
 \left|    A(\infty,  \infty,   \infty  ) -  A(Q ,\infty,\infty)   \right|   \leq & | A(\infty , \infty, \infty ) -  A(\infty,  n,   R  )| + |    A(\infty,  n,   R  )  -A(Q,  n ,  R)  |
\notag \\   & +  | A(Q,n ,R) -   A(Q ,\infty,\infty)  |    < \epsilon .
\end{align}
This proves that $  \lim_{Q \to \infty     }    A(Q ,\infty,\infty)  =    A(\infty,  \infty,   \infty  )  $ and
completes the proof of the second part of the lemma.
\end{proof}

\begin{remark}
The absolute value of the  integrand in the definition of  $A(Q, n, R)$ in Lemma \ref{lemma:uniform} is
\begin{align}\label{Nec}
 |\zeta(s) |
   \Big | \left\langle    \sigma_1 \Psi^{\overline \theta}_{\lambda_0},\left( H^\theta-z
    \right)^{-1} \sigma_1 \Psi^\theta_{\lambda_0}\right\rangle \Big |,
    \end{align}
and since the norm of the  resolvent operator behaves as $\Big |  1/z \Big | $
for large $|z|$, it is expected that the integral of \eqref{Nec} over $ \Gamma_-(\epsilon_n,R) $ diverges as $R$ tends to infinity.  A uniform bound of the from
\eqref{Misto} is possible because the oscillatory factor $  e^{-is (z - \lambda_0)} $is being integrated: we treat $ A(Q, n, R)  $ as an oscillatory integral, and use the usual tools from this area (we use a clever division of the   integration domain, apply integration by parts in different forms and interchange orders of integration). This is only possible if the variable $s$ is integrated (otherwise we loose the power of the oscillatory factor and we cannot perform integration by parts in the way we do). This is the reason why do not differentiate with respect to $Q$ and utilize the fundamental theorem of calculus (which is called Cook method in the context of scattering theory), since the derivative of $  A(Q, n, R) $ with respect to $Q$ does not contain an integration with respect to $s$.
\end{remark}

\subsection{Proof of Theorem~\ref{FK}}
\label{sec:proof-mainthm}

\begin{proof}[Proof of Theorem \ref{FK}]
Let $h,l\in\mathfrak h_0$; see Definition \ref{defasymptop}.
    Recall  the definition of $W$ given in \eqref{def:W1st} and the form factor
$f$ in \eqref{eq:f}.
Thanks to the fact that $f\in \mathit C^\infty(\R^3\setminus \{ 0\},\C)$,
we find that
\begin{align}
\label{eq:hflf}
hf,lf, W\in\mathfrak{h}_0.
\end{align}
Theorem~\ref{intker}, i.e.,  Equation~\eqref{eq:Tprecise1}  together with \eqref{gsanni} yields
\begin{align}
T(h,l)
 &=-2\pi i g\norm{\Psi_{\lambda_0}}^{-2} \left\langle a_-(W) \sigma_1 \Psi_{\lambda_0}, \Psi_{\lambda_0}\right\rangle
=-2\pi i g\norm{\Psi_{\lambda_0}}^{-2} \left\langle [a_-(W) ,\sigma_1] \Psi_{\lambda_0}, \Psi_{\lambda_0}\right\rangle ,
\end{align}
and
furthermore, recalling that $\omega(k)=|k|$,
and \eqref{a_-},  we obtain that
\begin{align}
\label{eq:thl0}
T(h,l)
&=
-2\pi(ig)^2\norm{\Psi_{\lambda_0}}^{-2} \int_{-\infty}^0 \mathrm{d}s   \,
\overline{\langle W_s,f\rangle_2}
\left\langle \left[e^{isH}\sigma_1
e^{-isH},  \sigma_1\right]  \Psi_{\lambda_0}, \Psi_{\lambda_0}\right\rangle
\notag \\
&=2\pi  g^2\norm{\Psi_{\lambda_0}}^{-2} \int^{\infty}_0 \mathrm{d}s \,
    \langle f, W_{-s}\rangle_2
    \left\langle \left[e^{-isH}\sigma_1
e^{isH},  \sigma_1\right]  \Psi_{\lambda_0}, \Psi_{\lambda_0}\right\rangle
\notag \\
&=i g^2 \norm{\Psi_{\lambda_0}}^{-2} \left(  T^{(1)}-T^{(2)} \right)
,
\end{align}
where  we use the abbreviations
\begin{align}
\label{eq:abr.terms}
T^{(j)}:=\lim\limits_{q\to 0^+} \lim\limits_{Q \to \infty}  T^{(j), q, Q}
\end{align}
 for $j=1,2$ with
\begin{align}
\label{eq:first-term}
T^{(1), q, Q}:&=-2\pi i\int^{Q}_{q} \mathrm{d}s  \int \mathrm{d^3} k   \,
W(k)f(k) e^{is(|k|+\lambda_0)}  \left\langle   \sigma_1 \Psi_{\lambda_0},
e^{-isH} \sigma_1\Psi_{\lambda_0}\right\rangle
\notag \\
&=      -2\pi i           \int^{Q}_{q} \mathrm{d}s  \int \mathrm{d} r   \,
G(r) e^{is(r+\lambda_0)}  \left\langle   \sigma_1 \Psi_{\lambda_0},
e^{-isH} \sigma_1\Psi_{\lambda_0}\right\rangle
\end{align}
and
\begin{align}
T^{(2), q, Q} :=-2\pi i\int^{Q}_{q} \mathrm{d}s  \int \mathrm{d} r   \,
G(r)e^{is(r-\lambda_0)}     \left\langle \sigma_1  \Psi_{\lambda_0},
e^{isH}\sigma_1\Psi_{\lambda_0}\right\rangle .
\label{eq:second-term}
\end{align}
Here, we use the notation
        \begin{align}
            \label{eq:G-def}
            G: \R \to \C , \qquad r \mapsto G(r):=
            \begin{cases}
                \int \mathrm{d}\Sigma \mathrm{d}\Sigma' \,  r^4  \overline{h(r,\Sigma)} l(r,\Sigma') f(r)^2  \qquad &\text{for} \quad r\geq 0
                \\
                0 \quad &\text{for} \quad r<0 ,
            \end{cases}
        \end{align}
            where we write spherical coordinates  $k=(r,\Sigma)$ and
        $k'=(r',\Sigma')$    in  \eqref{intker}  and \eqref{def:W1st}
         recalling the definition of $W$  and that
        $f(k)\equiv f(|k|)$  only depends on the radial coordinate  $r=|k|$.
         Thanks to  \eqref{eq:hflf}, we observe
        \begin{align}
            \label{eq:G-compact}
            &G\in \mathit C^\infty_c(\R\setminus \{ 0 \},\C) \subset \mathcal
                      S(\R ,\C)  .
        \end{align}

\paragraph{Term $T^{(1), q, Q}$:}
\cite[Theorem 2.3]{bdh-res}
guarantees  that $\Psi_{\lambda_0}$, and therefore, also $\sigma_1
\Psi_{\lambda_0}$ is an analytic vector (see Definition \ref{def:dil}). As pointed out earlier,  for the ground state, we can take the set $\mathcal{S}$ to be a neighborhood of $0$ which allows us to apply  Lemma \ref{laplace} and find
\begin{align}
\label{eq:thl1}
 T^{(1), q, Q} &=    -     \int^{Q}_{q} \mathrm{d}s  \int \mathrm{d}r  \,G(r)
 e^{is(r+\lambda_0)} \int_{\Gamma(\epsilon_n,R)}  \mathrm{d}z\, e^{-isz}
 \left\langle    \sigma_1 \Psi^{\overline \theta}_{\lambda_0},\left( H^\theta-z
 \right)^{-1} \sigma_1 \Psi^\theta_{\lambda_0}\right\rangle
 .
\end{align}
Here, $\Gamma(\epsilon_n,R)=\Gamma_{-}(\epsilon_n,R)\cup \Gamma_{c}(\epsilon_n)\cup
\Gamma_{d}(R)$ is the contour defined in Lemma \ref{laplace}, i.e.,
\eqref{Gamma-parts}, for sufficiently large $R>0$ and  $n > 2$.
We split the term
\begin{align}
    \label{eq:thl11}
    T^{(1), q, Q}
    &=T_{\epsilon_n,R}^{(1), q, Q}+T_{\epsilon_n}^{(1), q, Q}+T_R^{(1), q, Q}
\end{align}
 according to the different contours parts, see
\eqref{Gamma-parts}, in the $\mathrm d z$-integrals:
\begin{align}
    T_{\epsilon_n,R}^{(1), q, Q}:&=
-  \int^{Q}_{q} \mathrm{d}s \, J(s)  \int_{\Gamma_-(\epsilon_n,R)}  \mathrm{d}z\, e^{-isz}
    \left\langle    \sigma_1 \Psi^{\overline \theta}_{\lambda_0},\left( H^\theta-z
    \right)^{-1} \sigma_1 \Psi^\theta_{\lambda_0}\right\rangle
    ,
    \label{eq:thl11epsR}
    \\
    T_{\epsilon_n}^{(1), q, Q}:&=
 -    \int^{Q}_{q}\, J(s)  \int_{\Gamma_c(\epsilon_n)}  \mathrm{d}z\, e^{-isz}
    \left\langle    \sigma_1 \Psi^{\overline \theta}_{\lambda_0},\left( H^\theta-z
    \right)^{-1} \sigma_1 \Psi^\theta_{\lambda_0}\right\rangle
    \label{eq:thl11eps}
   ,
    \\
    T_R^{(1), q, Q}:&=
 - \int^{Q}_{q} \mathrm{d}s \, J(s)   \int_{\Gamma_d(R)}  \mathrm{d}z\, e^{-isz}
    \left\langle    \sigma_1 \Psi^{\overline \theta}_{\lambda_0},\left( H^\theta-z
    \right)^{-1} \sigma_1 \Psi^\theta_{\lambda_0}\right\rangle ,
    \label{eq:thl11R}
\end{align}
and we use the definition
\begin{align}
\label{def:J}
J:  \R \to \C , \qquad s\mapsto J(s)= \int \mathrm{d} r  \,G(r) e^{is(r+\lambda_0)}.
\end{align}
We observe that, thanks to \eqref{eq:G-compact},  we have  $J\in \mathit S(\R,\C)$ which implies
       \begin{align}
    \left| J(s) \right|       &\leq C       (1 + |s|^2)^{-1}
            \label{eq:int-bound-a}
        \end{align}
        for some constant $C$.
Moreover, we have    (see \eqref{resA})
    \begin{align}
            \left| e^{-isz}\left\langle    \sigma_1 \Psi^{\overline
                    \theta}_{\lambda_0},\left( H^\theta-z
            \right)^{-1} \sigma_1 \Psi^\theta_{\lambda_0}\right\rangle\right|
            &\leq  C  \|\Psi_{\lambda_0}\|^2
            \frac{e^{s\Im z}}{| z - e_1|}
           , \qquad \forall z \in \Gamma_d(R).
            \label{eq:int-bound-b}
        \end{align}

\paragraph{Contribution $T_{\epsilon_n}^{(1), q, Q}$ in \eqref{eq:thl11eps}:}
Using \eqref{eq:int-bound-a}, we may start with
the bound
\begin{align}
\label{eq:finaleps}
 |T_{\epsilon_n}^{(1), q, Q}|
\leq C \sup_{s\in [q,Q]}
   \left|  \int_{\Gamma_c(\epsilon_n)}   \mathrm{d}z\, e^{-isz}
    \left\langle    \sigma_1 \Psi^{\overline \theta}_{\lambda_0},\left( H^\theta-z
    \right)^{-1} \sigma_1 \Psi^\theta_{\lambda_0}\right\rangle\right| .
\end{align}
It  follows  from \eqref{nn0}   together with Definition \ref{Defen}  that
there is a  constant $C$ such that, for $s\in [q,Q]$, we have
\begin{align}
    \left|\int_{\Gamma_c(\epsilon_n)}  \mathrm{d}z\, e^{-isz}
      \left\langle
        \sigma_1 \Psi^{\overline \theta}_{\lambda_0},\left( H^\theta-z  \right)^{-1}
    \sigma_1 \Psi^\theta_{\lambda_0}\right\rangle\right| \leq Ce^{\epsilon_n Q}
    \frac{\epsilon_n}{\rho_n}  \boldsymbol{C}^{n+1} \leq C e^{\epsilon_n Q} \rho_n^{\mu/ 8} ,
\end{align}
where we use \eqref{dorm2}.
In conclusion, we have that, for all $0<q<Q<\infty$,
\begin{align}
\label{eq:contr1}
\lim\limits_{n\to 0} T_{\epsilon_n}^{(1), q,Q}=0.
\end{align}

\paragraph{Contribution $T_R^{(1),q,Q}$ in \eqref{eq:thl11R}:}
Using \eqref{eq:int-bound-a} again, we find
\begin{align}
|T_{R}^{(1),q,Q}| \leq C \int_{q}^{Q}\mathrm{d}s\, \frac{1}{1 + |s|^2}
   \left| \int_{\Gamma_d(R)}  \mathrm{d}z\,  e^{-isz}
    \left\langle    \sigma_1 \Psi^{\overline \theta}_{\lambda_0},\left( H^\theta-z
    \right)^{-1} \sigma_1 \Psi^\theta_{\lambda_0}\right\rangle\right|
    .
    \label{eq:lest}
\end{align}
For $s\in [q,Q]$, we observe that there is a constant $C$ such that      (see \eqref{resA})
\begin{align} \label{aver}
\left| \int_{\Gamma_d(R)}  \mathrm{d}z\, e^{-isz}  \left\langle    \sigma_1 \Psi^{\overline \theta}_{\lambda_0},\left( H^\theta-z  \right)^{-1} \sigma_1 \Psi^\theta_{\lambda_0}\right\rangle\right|
\leq \frac{C}{R} \int_0^\infty \mathrm{d}u \,  e^{-su \sin
    (\nu/4)} .
\end{align}
     Thereby, as in \eqref{aver}, we obtain the estimate
\begin{align}\label{aver1}
    &\lim_{R \to \infty} \int_{q}^{Q}\mathrm{d}s\, \frac{1}{1 + |s|^2}
    \int_{\Gamma_d(R)}  \mathrm{d}z\, \left| e^{-isz}
    \left\langle    \sigma_1 \Psi^{\overline \theta}_{\lambda_0},\left( H^\theta-z
    \right)^{-1} \sigma_1 \Psi^\theta_{\lambda_0}\right\rangle\right|
    \notag\\
    &\leq
   \lim_{R \to \infty } \frac{C}{R} \int_{q}^{Q}\mathrm{d}s\, \frac{1}{1 + |s|^2}\frac{1}{|s|} = 0.
\end{align}
 Then, we conclude   for all $0<q<Q<\infty$
 \begin{align}
 \label{eq:contr2}
 \lim\limits_{R\to\infty}  T_R^{(1),q,Q} =0 .
\end{align}
This together with \eqref{eq:contr1} and \eqref{eq:thl11} yields that for all $0<q<Q<\infty$
\begin{align}
T^{(1), q, Q} =\lim\limits_{n,R\to\infty}T_{\epsilon_n,R}^{(1), q, Q} .
\end{align}
Note that $J\in\mathit S(\R,\C)$.  Therefore, we are in the position to apply
Lemma \ref{lemma:uniform} and find
\begin{align}
\label{fincontr1}
T^{(1), q, \infty} :=\lim\limits_{Q\to\infty}T^{(1), q, Q}  =    \lim_{Q\to \infty} \lim_{  n, R   \to \infty }  T_{\epsilon_n,R}^{(1), q, Q }        =\lim\limits_{n,R\to\infty}T_{\epsilon_n,R}^{(1), q, \infty} ,
\end{align}
where
\begin{align}
& T_{\epsilon_n,R}^{(1), q, \infty}:=\lim\limits_{Q\to\infty}T_{\epsilon_n,R}^{(1), q, Q}
= - \int^{\infty}_{q} \mathrm{d}s  \, J(s)   \int_{\Gamma_-(\epsilon_n,R)}  \mathrm{d}z\, e^{-isz}
    \left\langle    \sigma_1 \Psi^{\overline \theta}_{\lambda_0},\left( H^\theta-z
    \right)^{-1} \sigma_1 \Psi^\theta_{\lambda_0}\right\rangle .
\end{align}
 For fixed $n$ and $R$, the function $z   \mapsto e^{-isz}
    \left\langle    \sigma_1 \Psi^{\overline \theta}_{\lambda_0},\left( H^\theta-z
    \right)^{-1} \sigma_1 \Psi^\theta_{\lambda_0}\right\rangle$ is bounded   in $   \Gamma_-(\epsilon_n,R)  $.
 Then, thanks to \eqref{eq:int-bound-a}, we may apply Fubini's theorem and find:
 \begin{align}
 \label{eq.thl1-final1}
& T_{\epsilon_n,R}^{(1), q, \infty}
=-\int_{\Gamma_-(\epsilon_n,R)}  \mathrm{d}z\,
    \left\langle    \sigma_1 \Psi^{\overline \theta}_{\lambda_0},\left( H^\theta-z
    \right)^{-1} \sigma_1 \Psi^\theta_{\lambda_0}\right\rangle  \int^{\infty}_{q} \mathrm{d}s  \int \mathrm{d} r  \,G(r)
    e^{is(r+\lambda_0-z)}
    \notag \\
    &=-\int_{\Gamma_-(\epsilon_n,R)}  \mathrm{d}z\,
    \left\langle    \sigma_1 \Psi^{\overline \theta}_{\lambda_0},\left( H^\theta-z
    \right)^{-1} \sigma_1 \Psi^\theta_{\lambda_0}\right\rangle  \int \mathrm{d}s \, \Theta (s-q)  \int \mathrm{d} r  \,G^{(z)}(r)
    e^{-isr}  .
\end{align}
In the last step, we  use the coordinate transformation $r\to z-\lambda_0-r$ and
  the notation
        \begin{align}
            \label{eq:Gz-def}
            G^{(z)}: \R \to \C , \qquad r \mapsto G^{(z)}(r):=G(z-\lambda_0-r) \qquad z\in \R.
        \end{align}
        Then, it follows from  \eqref{eq:G-compact}
        together with   \eqref{def:fourier-distri} that
         \begin{align}
\int \mathrm{d}s \, \Theta (s-q)  \int \mathrm{d} r  \,G^{(z)}(r)  e^{-isr} &=\int \mathrm{d}s \, \Theta (s)  \int \mathrm{d} r  \,G^{(z)}(r) e^{-iqr} e^{-isr}
\notag \\
&=\Theta(\mathfrak{F}[G^{(z),q}])=     \mathfrak{F}[     \Theta](      G^{(z),q}) ,
\end{align}
 where, for $q>0$, we define
 \begin{align}
G^{(z),q}(r):=G^{(z)}(r) e^{-iqr} .
 \end{align}
 Thanks to \eqref{eq:G-compact}, we have  for  $z\in\R$ and $q\geq 0$
        \begin{align}
            &G^{(z),q}\in C^\infty_c    (\R\setminus \{    z -  \lambda_0 \},\C) \subset \mathcal
            S(\R      ,\C)   .
        \end{align}
 It follows from Lemma \ref{lemma:heaviside} that for  $z\in\R$
   \begin{align}
\int \mathrm{d}s \, \Theta (s-q)  \int \mathrm{d} r  \,G^{(z)}(r)  e^{-isr}
&=\pi \delta(G^{(z),q}) -i \left(\text{PV}\left(1/\bullet \right)\right)(G^{(z),q})  .
\end{align}
This together with \eqref{eq.thl1-final1} yields that
\begin{align}
\label{eq:thl-1}
T_{\epsilon_n,R}^{(1),q,\infty}
&= T_{\epsilon_n,R}^{(1,1),q,\infty}+T_{\epsilon_n,R}^{(1,2),q,\infty},
\end{align}
where
\begin{align}
T_{\epsilon_n,R}^{(1,1),q,\infty}:&=-\pi \int_{\Gamma_-(\epsilon_n,R)}  \mathrm{d}z\,   \left\langle    \sigma_1 \Psi^{\overline \theta}_{\lambda_0},\left( H^\theta-z  \right)^{-1} \sigma_1 \Psi^\theta_{\lambda_0}\right\rangle  G(z-\lambda_0)
\\
T_{\epsilon_n,R}^{(1,2),q,\infty}:&=i \int_{\Gamma_-(\epsilon_n,R)}  \mathrm{d}z\,   \left\langle
\sigma_1 \Psi^{\overline \theta}_{\lambda_0},\left( H^\theta-z  \right)^{-1}
\sigma_1 \Psi^\theta_{\lambda_0}\right\rangle  \lim\limits_{\eta\to 0^+}
\int_{\R\setminus [-\eta,\eta]}\mathrm{d}r \, \frac{G(z-\lambda_0-r)e^{-iqr}}{r}
\label{eq:t12-lim}
\end{align}
In the following, we shall compute both contributions explicitly.

\paragraph{Contribution $T_{\epsilon_n,R}^{(1,1)}(h,l)$:}
It follows from \eqref{eq:G-compact} that there are numbers  $M>\kappa>0$ such that
$\text{supp }  G \subset [\kappa,M]$.   Recall that everything so far
holds for any choice of $n,R>0$  large enough.
For the rest of this proof we will restrict this choice to    $R > M $   and $n>0$ large enough such that
 $\epsilon_n<\kappa/4$.
In this setting, we may turn the $\mathrm dz$-integral in an indefinite one,
exploiting, the compact support of $G$ and the definition of the contour
$\Gamma_-(\epsilon_n,R)$. We thus obtain
\begin{align}
T_{\epsilon_n,R}^{(1,1),q,\infty}
&=-\pi \int_{\Gamma_-(\epsilon_n,R)}  \mathrm{d}z\,   \left\langle    \sigma_1 \Psi^{\overline \theta}_{\lambda_0},\left( H^\theta-z  \right)^{-1} \sigma_1 \Psi^\theta_{\lambda_0}\right\rangle  G(z-\lambda_0)
\notag \\
&=-\pi \int_{\Gamma_-(\epsilon_n,R)-\lambda_0}  \mathrm{d}z\,   \left\langle    \sigma_1 \Psi^{\overline \theta}_{\lambda_0},\left( H^\theta-\lambda_0-z  \right)^{-1} \sigma_1 \Psi^\theta_{\lambda_0}\right\rangle  G(z)
 \notag \\
&=-\pi \int^\infty_{0}  \mathrm{d}z\,   \left\langle    \sigma_1 \Psi^{\overline \theta}_{\lambda_0},\left( H^\theta-\lambda_0-z  \right)^{-1} \sigma_1 \Psi^\theta_{\lambda_0}\right\rangle  G(z)
\label{eq:t11-G}
\end{align}

\paragraph{Contribution $T_{\epsilon_n,R}^{(1,2)}(h,l)$:}
In order to calculate $T_{\epsilon_n,R}^{(1,2)}(h,l)$
we can now fall back to
Lemma \ref{lemma:thl12}.
 We recall  Definition \ref{Defen}
 and notice that $0<\epsilon_n <\kappa/4$ for sufficiently large $n$. Then,
as a direct consequence of Lemma \ref{lemma:thl12}, we find (for sufficiently large $R$)
\begin{align}
\lim_{n,R \to \infty }T_{\epsilon_n,R}^{(1,2),q,\infty}&= i  \lim\limits_{n,R\to\infty, \eta \to 0}T_{n,R}(\eta)
 \notag
 \\
 &= -\pi
\int_{\R}\mathrm{d}r \, G(r)e^{-iqr}\left\langle    \sigma_1 \Psi^{\overline
\theta}_{\lambda_0},\left( H^\theta-\lambda_0-r \right)^{-1} \sigma_1
\Psi^\theta_{\lambda_0}\right\rangle
\notag\\
&= -\pi
\int_0^\infty\mathrm{d}z \, \left\langle    \sigma_1 \Psi^{\overline
\theta}_{\lambda_0},\left( H^\theta-\lambda_0- z \right)^{-1} \sigma_1
\Psi^\theta_{\lambda_0}\right\rangle G(z)e^{-iqz},
\label{eq:t11-PV}
\end{align}
where $T_{n,R}(\eta)$ is defined in \eqref{eq:t_eps,R,eta}.

Collecting the contributions of
\eqref{eq:thl-1}, i.e,
\eqref{eq:t11-G} and \eqref{eq:t11-PV},
we establish the identity
\begin{align}
\label{eq:thl1-final}
T^{(1)}&=\lim\limits_{q\to 0^+}\lim_{n,R \to \infty} T_{\epsilon_n,R}^{(1),q,\infty}
 \\\notag
&=
-\pi \lim\limits_{q\to 0^+} \int^\infty_{0}  \mathrm{d}z\,   \left\langle    \sigma_1 \Psi^{\overline \theta}_{\lambda_0},\left( H^\theta-\lambda_0-z  \right)^{-1} \sigma_1 \Psi^\theta_{\lambda_0}\right\rangle  G(z)(1+e^{-iqz})
 \\\notag
&=
-2\pi \int^\infty_{0} \mathrm{d}z\,   \left\langle    \sigma_1 \Psi^{\overline \theta}_{\lambda_0},\left( H^\theta-\lambda_0-z  \right)^{-1} \sigma_1 \Psi^\theta_{\lambda_0}\right\rangle  G(z)
 \\\notag
&=-2\pi \int \mathrm{d}^3k \mathrm{d}^3k' \,  \overline{h(k)} l(k') f(k)f(k') \delta(|k|-|k'|)  \left\langle    \sigma_1 \Psi^{\overline \theta}_{\lambda_0},\left( H^\theta-\lambda_0-|k'| \right)^{-1} \sigma_1 \Psi^\theta_{\lambda_0}\right\rangle    .
\end{align}
In the third line we applied the dominated convergence theorem which is justified by \eqref{eq:G-compact}.
Moreover, we have inserted the definition of $G$  using the symbolic notation of
the Dirac-delta distribution in the last step.

\paragraph{Term $T^{(2)}$:}
The second term  $T^{(2)}$ can  be inferred
by repeating the calculation with $\theta$ replaced by $\overline\theta$ and reflecting the
path of integration $\Gamma (\epsilon_n,R)$ on the real axis when applying Lemma
\ref{laplace}. In this case one  has to consider  the Hamiltonian
$H^{\overline \theta}$ whose spectrum is  given by mirroring the spectrum
of $H^\theta$ at the real axis.
Due to the similarity of the calculation, we omit a proof but only state the
result
\begin{align}
\label{eq:proofmain-thl2final}
&T^{(2)}
=2\pi   \int \mathrm{d}^3k \mathrm{d}^3k' \,  \overline{h(k)} l(k') f(k)f(k') \delta(|k|-|k'|)  \left\langle    \sigma_1 \Psi^{ \theta}_{\lambda_0},\left( H^{\overline \theta}-\lambda_0+|k'| \right)^{-1} \sigma_1 \Psi^{\overline \theta}_{\lambda_0}\right\rangle .
\end{align}
The relative sign in comparison with \eqref{eq:thl1-final} is due to the the opposite mathematical orientation of the contour.
Inserting \eqref{eq:thl1-final} and
\eqref{eq:proofmain-thl2final} in \eqref{eq:thl0} completes the proof.
\end{proof}

\begin{appendix}
\section{Collection of previous results used in this work}
\label{app:previous}

In this section we collect the relevant results of \cite{bdh-scat} and \cite{bdh-res}
 which are used in the proofs contained in this work.

\subsection{Scattering Theory}
Let $\Psi\in \mathcal K \otimes D(H_f^{1/2})$ and $h,l\in\mathfrak h_0$.  Then, we recall from \cite[Lemma 4.1]{bdh-scat} that
        \begin{align}
            \label{a_-}
            a_-(h)\Psi=a(h)\Psi + ig\int^0_{-\infty} \mathrm{d}s\,
        e^{isH}
        \langle h_s,f\rangle_2\,
            \sigma_1
            e^{-isH} \Psi.
        \end{align}
        It can be shown by  integration by parts   that there is constant $C$ such that $| \langle h_s,f\rangle_2 |\leq C/(1+s^2)$ for $s\in\R$ (see \cite[Eq.\ (C.7)]{bdh-scat}). Hence, the integral  above is convergent.
Moreover, it is proven in \cite[Lemma 4.1 (iv)]{bdh-scat} that
\begin{align}
  a_\pm(h)\Psi_{\lambda_0}=0 . \label{gsanni}
\end{align}

\subsection{Spectral Properties}
We define
\begin{align}
\label{def:setS}
 \mathcal S:=\left\{\theta\in\C: -10^{-3} <   \Re \theta < 10^{-3} \text{ and }
\boldsymbol{\nu} < \Im \theta  < \pi/16 \right\} ,
\end{align}
 where $\boldsymbol \nu \in (0, \pi/16)$ is a fixed number (see \cite[Definition 1.4]{bdh-res}).

 In order to specify some of the spectral properties of $
 H^{\theta} $ we  define certain regions in the complex plane:
 \begin{definition}{(c.f.\ \cite[Definition 3.2]{bdh-scat})}
 \label{def:regionsAB}
 For fixed $\theta\in \mathcal S$, we set $\delta = e_1- e_0 = e_1$ and define the regions
 \begin{align}
 \label{region:A}
 A:&=
 A_1\cup A_2\cup A_3  ,
 \end{align}
 where
 \begin{align}
 A_1:&=\left\{  z\in\C : \Re z <e_0-\delta/2 \right\}
 \\
 A_2:&= \left\{  z\in\C : \Im z >\frac{1}{8}\delta \sin (\nu) \right\}
 \\
 A_3:&= \left\{  z\in\C : \Re z >e_1+\delta/2 , \Im z \geq -\sin (\nu/2) \left(\Re (z) -(e_1+\delta/2)   \right)\right\} ,
 \end{align}
 and for $i=0,1$, we define
 \begin{align}
 \label{region:Bi1}
 B_i^{(1)}:=\left\{  z\in\C : |\Re z-e_i| \leq \frac{1}{2}\delta, -\frac{1}{2}
 \rho_1 \sin(\nu)\leq \Im z \leq \frac{1}{8}\delta \sin (\nu)  \right\} .
 \end{align}
 These regions are depicted in Figure \ref{fig:regionsAB}.
 \end{definition}
 \begin{figure}[h]
 \centering
 \includegraphics[width=\textwidth]{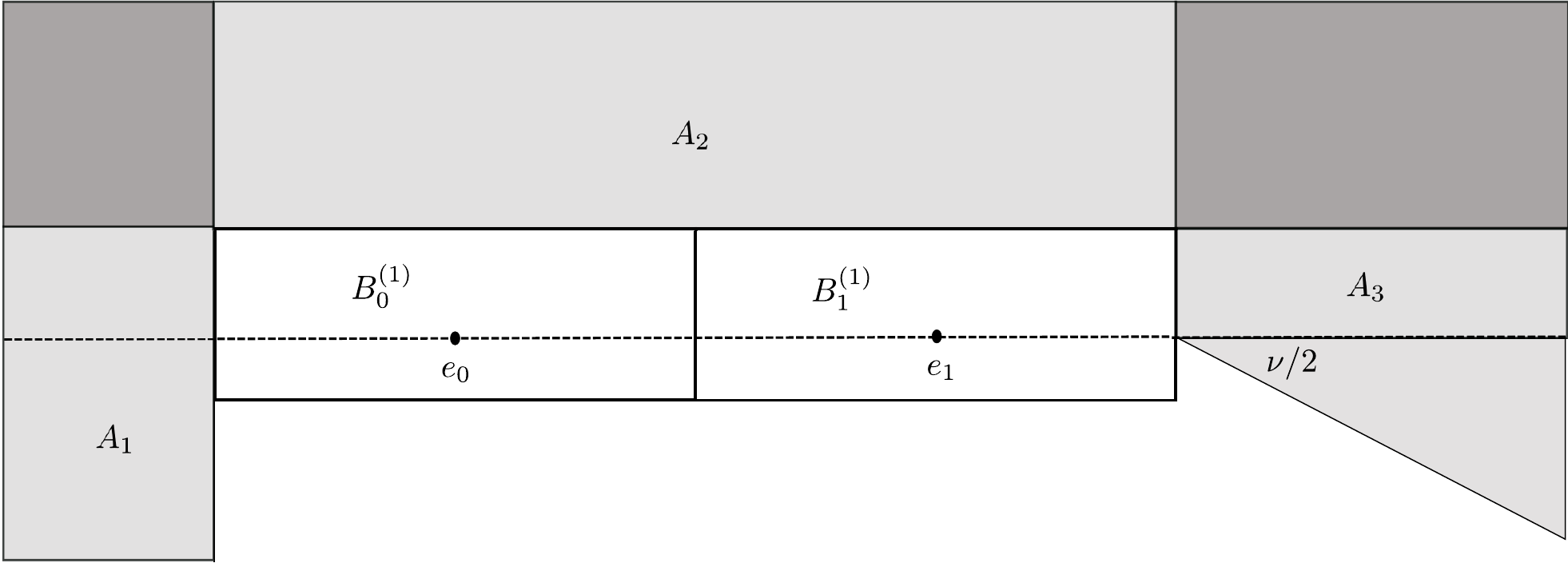}
   \caption{An illustration of the subsets of the complex plane introduced in
   Definition \ref{def:regionsAB}.}
     \label{fig:regionsAB}
 \end{figure}
 For a fixed  $ m \in \mathbb{N}, \:  m \geq 4,$ we define the cone
\begin{align}
\label{eq:defcone}
\mathcal C_m(z) :=\left\{  z+xe^{-i\alpha} : x\geq 0 ,
|\alpha-\nu |\leq  \nu/m \right\} .
\end{align}
 It follows from the induction scheme in \cite[Section 4]{bdh-res} that $ \lambda_i \in  B_i^{(1)} $, and moreover,
  \cite[Theorem 2.7]{bdh-res}  together with \cite[Lemma 3.13]{bdh-res}  yields
 \begin{align}\label{spectrum}
 \sigma(H^{\theta}) \subset \mathbb{C} \setminus  \Big [  A \cup \big ( B_0^{(1)} \setminus  \mathcal C_m(\lambda_0)   \big ) \cup  \big ( B_1^{(1)} \setminus  \mathcal C_m(\lambda_1)   \big )      \Big ].
 \end{align}
For $g$ small enough, we recall from \cite[Eq.\ (3.13)]{bdh-scat} that there is   constant  $\boldsymbol{c}>0$ such that
 \begin{align}
 \label{eq:impartres}
 \Im \lambda_1 <-g^2  \boldsymbol{c} <0 .
 \end{align}
In the following we collect some important resolvent estimates.
 The region $A$ is far away from the spectrum, and therefore,  resolvent estimates in this  region are easy. In  \cite[Lemma 3.2]{bdh-res},  we prove that there is a constant $C$ (that does not depend on $ n, g, \rho_0 $ and $\rho$) such that
\begin{align}\label{resA}
\Big \| \frac{1}{H^{\theta} - z}  \Big \| \leq C \frac{1}{ |z  - e_1|}, \qquad
\forall z \in A.
\end{align}
As in \cite[Eq.\ (3.31)]{bdh-scat}, we select  the auxiliary numbers $ \rho$
\begin{align}\label{dorm2}
\boldsymbol{C}^8 \rho_0^{\mu} \leq 1, \qquad \boldsymbol{C}^8 \rho^{\mu} \leq 1/4 ,  \qquad   (\text{and hence} \qquad  \boldsymbol{C} \rho^{ \frac{1}{2}  \iota (1 + \mu/4)} \leq 1),
\end{align}
where
\begin{align}\label{dorm222}
\iota = \frac{\mu/4 }{ ( 1+ \mu/4 ) } \in (0, 1).
\end{align}
In \cite[Lemma 4.7]{bdh-scat} we show that for all $n\in\N$,
a fixed (arbitrary)  $m\geq 4$
and $\theta\in\mathcal S$, there is a constant   $ C $
(that depends on $m$)
such that
\begin{align}\label{nn0}
\Big \| \frac{1}{H^{\theta}-z} \sigma_1
\Psi_{\lambda_0}^{\theta}    \Big \| \leq  C \boldsymbol{C}^{n+1} \frac{1}{\rho_n},
\end{align}
for every $ z \in  B_0^{(1)} \setminus  \mathcal{C}_{m}( \lambda_0 - 2 \rho_n^{1+ \mu/ 4} e^{-i \nu}) $, where the cone $ \mathcal{C}_{m}$ is defined in \eqref{eq:defcone}. It can be seen from  \cite[Lemma 4.7]{bdh-scat}  that $C$ does not depend on $n$, $\rho_0$ and $\rho$.
Here, we recall from \cite[Eq.\ (4.51)]{bdh-scat} that
\begin{align}\label{caritaputa}
\mathcal C_m ( \lambda_0 - 2 \rho_{n}^{1+ \mu/4}  e^{-i \nu } ) \cap
 \Big ( \overline{\mathbb{C}^+} + \lambda_0 - i 2 \sin(\nu) \rho_n^{1 + \mu /4 } \Big ) \subset    D(\lambda_0,   \epsilon_n )  \subset
 D(\lambda_0,  2 \epsilon_n ) \subset B_0^{(1)} .
\end{align}

\section{Proof of Lemma \ref{lemma:heaviside}}
\label{app:heaviside}

\begin{proof}[Proof of Lemma \ref{lemma:heaviside}]
For $\alpha>0$,  we define $g_\alpha \in \mathit S'(\R,\C)$ by
\begin{align}
g_\alpha : \mathit S(\R,\C) \to \C , \qquad
\varphi \mapsto g_\alpha (\varphi)=\int_0^\infty \mathrm{d}x\, e^{-\alpha x} \varphi(x) .
\end{align}
It follows from \eqref{def:fourier-distri} that for $\varphi\in \mathit S(\R,\C)$
\begin{align}
\label{E2}
\mathfrak F[g_\alpha] (\varphi)=g_\alpha\left( \mathfrak F[\varphi]\right)= \int^\infty_0 \mathrm{d}x\, e^{-\alpha x} \mathfrak F[\varphi](x)
=\int^\infty_0 \mathrm{d}x\, e^{-\alpha x} \int_\R \mathrm{d}s \,  \varphi(s)e^{-isx}.
\end{align}
The  integrand on the right-hand side of \eqref{E2} is absolutely integrable because of $\varphi\in \mathit S(\R,\C)$, and hence,  the Fubini-Tonelli theorem yields that
\begin{align}
\mathfrak F[g_\alpha] (\varphi)
= \int_\R \mathrm{d}s \,  \varphi(s) \int^\infty_0 \mathrm{d}x\, e^{-x(\alpha+is)}.
\end{align}
This together with
\begin{align}
\int^\infty_0 \mathrm{d}x\, e^{-x(\alpha+is)} = \frac{1}{\alpha +is}= \frac{\alpha}{(\alpha^2+s^2)}-i \frac{s}{(\alpha^2+s^2)}
\end{align}
implies that
\begin{align}
\label{eq:lim10001}
\mathfrak{F}[g_\alpha](\varphi)&=G^{(1)}_\alpha(\varphi)-iG^{(2)}_\alpha(\varphi) ,
\end{align}
where
\begin{align}
G^{(1)}_\alpha(\varphi)
&=\int_\R \mathrm{d}s\,\frac{\alpha}{(\alpha^2+s^2)}\varphi(s)
\end{align}
and
\begin{align}
G^{(2)}_\alpha(\varphi)
&=\int_\R \mathrm{d}s\,\frac{s}{(\alpha^2+s^2)}\varphi(s) .
\end{align}
Using the coordinate transformation $s\to \alpha s$ we obtain that
\begin{align}
\label{eq:heavi-lim1}
\lim\limits_{\alpha\to 0^+} G^{(1)}_\alpha(\varphi)
= \lim\limits_{\alpha\to 0^+} \int_\R  \mathrm{d}s\, \frac{\varphi(\alpha s) }{1+s^2}
=\varphi(0)\int_\R  \mathrm{d}s\, \frac{1}{1+s^2}=\pi \varphi(0)=\pi \delta (\varphi) ,
\end{align}
where the second step follows from the dominated convergence theorem together with the continuity of $\varphi$.
Moreover, we have
\begin{align}
\label{eq:heavi-lim20}
G^{(2)}_\alpha(\varphi)
=G^{(2,1)}_\alpha(\varphi)+G^{(2,2)}_\alpha(\varphi) ,
\end{align}
where
\begin{align}
G^{(2,1)}_\alpha(\varphi):=\int_{\R\setminus [-\alpha^8,\alpha^8]} \mathrm{d}s\,\frac{s}{(\alpha^2+s^2)}\varphi(s)
\end{align}
and
\begin{align}
G^{(2,2)}_\alpha(\varphi):=\int^{\alpha^8}_{-\alpha^8} \mathrm{d}s\,\frac{s}{(\alpha^2+s^2)}\varphi(s).
\end{align}
We treat these two terms separately. At first, we obtain
\begin{align}
\left| G^{(2,2)}_\alpha(\varphi) \right|
&\leq \int^{\alpha^8}_{-\alpha^8} \mathrm{d}s\,\left| \frac{s}{(\alpha^2+s^2)}\left( \varphi(s)-\varphi(0)\right)\right|
+|\varphi(0)| \left| \int^{\alpha^8}_{-\alpha^8} \mathrm{d}s\, \frac{s}{(\alpha^2+s^2)}\right|
\notag \\
&\leq 2\alpha^{14}  \sup_{s\in [-\alpha^8,\alpha^8]}\left| \varphi(s)-\varphi(0)\right|
+\frac{|\varphi(0)|}{2}\left| \int^{\alpha^{16}}_{-\alpha^{16}} \mathrm{d}s\, \frac{1}{\alpha^2+s}\right|
\end{align}
where we have used the coordinate transformation $s'=s^2$ for the second term in the last line. Then, we obtain
\begin{align}
\left| G^{(2,2)}_\alpha(\varphi) \right|
&\leq  2\alpha^{14}  \sup_{s\in [-\alpha^8,\alpha^8]}\left| \varphi(s)-\varphi(0)\right|
+\frac{\varphi(0)}{2}\left| \ln (1+\alpha^8)-\ln (1-\alpha^8)\right| .
\end{align}
Note that $\ln(\cdot)$ is continuous close to $1$ and $\sup_{s\in [-\alpha^8,\alpha^8]}\left| \varphi(s)-\varphi(0)\right|<\infty$ since a continuous function has a maximum on a compact set.  We conclude
\begin{align}
\label{eq:heavi-lim22}
\lim\limits_{\alpha\to 0^+} G^{(2,2)}_\alpha(\varphi)
=0.
\end{align}
Finally, for some $R>0$, we obtain
\begin{align}
G^{(2,1)}_\alpha(\varphi)
&= \int_{[-R,R]\setminus [-\alpha^8,\alpha^8]} \mathrm{d}s\,\frac{s}{(\alpha^2+s^2)}\left( \varphi(s) -\varphi (0)\right)
+\int_{[-R,R]\setminus [-\alpha^8,\alpha^8]} \mathrm{d}s\,\frac{s}{(\alpha^2+s^2)}\varphi (0)
\notag \\
&+\int_{\R\setminus [-R,R]} \mathrm{d}s\,\frac{s}{(\alpha^2+s^2)}\varphi (s) .
\end{align}
Due to symmetry, the second term vanishes independently of $R$, and moreover,
 the mean value theorem implies that
\begin{align}
|\varphi(s) -\varphi (0)| \leq |s| \norm{\varphi'}_\infty .
\end{align}
Altogether, this yields that
\begin{align}
\Big | \frac{s}{(\alpha^2+s^2)}\left( \varphi(s) -\varphi (0)\right)  \chi_{   [-R,R]\setminus [-\alpha^8,\alpha^8]     }(s) \Big |
\leq  &   \norm{\varphi'}_\infty \chi_{   [-R,R]   }(s),   \\
\Big | \frac{s}{(\alpha^2+s^2)} \varphi (s)  \chi_{ \R \setminus   [-R,R]     }(s) \Big |
\leq  &     \Big |   \frac{\phi(s)}{s}     \chi_{ \R \setminus   [-R,R]     }(s)   \Big |,
\end{align}
where $\chi_A$ is the characteristic (indicator) function of the set $A$.
This allows us to apply the dominated convergence theorem in order to find
\begin{align}
\lim\limits_{\alpha\to 0^+}  G^{(2,1)}_\alpha(\varphi)
&= \text{PV} \int_{\R} \mathrm{d}s\,\frac{1}{s} \varphi(s)=\left( \text{PV}\left(1/ \bullet \right)\right)(\varphi) .
\end{align}
This together with \eqref{eq:heavi-lim22},  \eqref{eq:heavi-lim20}, \eqref{eq:heavi-lim1} and \eqref{eq:lim10001} implies that
\begin{align}
\label{eq:distr1}
\lim\limits_{\alpha\to 0^+} \mathfrak{F}[g_\alpha](\varphi)&=\pi\delta
(\varphi)-i\left(\text{PV}\left(1/\bullet \right)\right)(\varphi)
\qquad \forall\varphi\in \mathit S(\R,\C) .
\end{align}
We conclude the proof by \eqref{def:fourier-distri} which yields
\begin{align}
\label{eq:heavi-lim0}
\lim\limits_{\alpha\to 0^+}\mathfrak{F}[g_\alpha](\varphi)
=\lim\limits_{\alpha\to 0^+}g_\alpha(\mathfrak{F}[\varphi])
=\Theta(\mathfrak{F}[\varphi])
=\mathfrak{F}[\Theta](\varphi)  \qquad \forall\varphi\in \mathit S(\R,\C) .
\end{align}
\end{proof}
\end{appendix}

\section*{Acknowledgement}
D.\ -A.\ Deckert and F.\ H\"anle would like to thank the IIMAS at UNAM and M.\
Ballesteros  the Mathematisches Institut at LMU Munich   for their hospitality. This
project was partially funded by the DFG Grant DE 1474/3-1,  the grants PAPIIT-DGAPA
UNAM  IN108818, SEP-CONACYT 254062, and the junior research group ``Interaction
between Light and Matter'' of the Elite Network Bavaria. M.\  B.\  is a
Fellow of the Sistema Nacional de Investigadores (SNI).
F.\ H.\ gratefully acknowledges financial support by the ``Studienstiftung des deutschen Volkes''.
Moreover, the authors express their gratitude for the fruitful discussions with
V.\ Bach, J.\ Faupin, J.\ S.\ M\o ller, A.\ Pizzo and W.\ De Roeck, R. Weder and P. Barberis.

\bibliographystyle{amsplain}
\bibliography{ref}
\end{document}